%% file: main.tex
\PassOptionsToPackage{dvipsnames}{xcolor}
\documentclass[11pt, thm-restate, cleveref, letterpaper]{article}
\newif\ifdagstuhl
\dagstuhlfalse

\ifdagstuhl
\input{lipicsheader}
\else
\input{arxivheader}
\fi

\begin{document}
\maketitle
\ifdagstuhl

\else 
\thispagestyle{empty}

\fi

\begin{abstract}
{We develop low distortion embeddings with outliers from arbitrary metrics into hierarchically separated trees (HSTs). 
In particular, we develop an efficient algorithm that for any $\epsilon>0$, given an input metric $(X,d)$, and a probabilistic embedding of all but $k$ points from $X$ into HSTs with distortion $c$, samples from a probabilistic embedding of all but $O(\frac{k}{\epsilon}\log k)$ points into HSTs that achieves distortion at most $(32+\epsilon)c$.

Our results are based on two key technical components. First, we extend an algorithm of Munagala et al. \cite{mungala2023} for minimizing the distortion of embeddings without outliers into HSTs to the setting with outliers. We combine this with new results on bi-Lipschitz extensions into trees and $\ell_1$ space. In particular, we show that any probabilistic embedding into HSTs can be extended to $k$ additional points with only a factor $O(\log k)$ of additional distortion. This \bl extension result utilizes a new probabilistic partitioning scheme that we call {\em onion partitioning}.}
\end{abstract}

\ifdagstuhl
\else 
\clearpage

\setcounter{page}{1}

\fi

\input{intro}

\input{definitions}

\input{components}

\input{ckr_background}

\input{merge_functions}

\input{hst}

\input{applications}

\section{Open questions}

There are still many open questions we can ask about topics in this space, such as:

\begin{itemize}
    \item Is it NP-hard to find the optimal outlier set size when probabilistically embedding a given metric into HSTs with a small fixed distortion? This question is known to be hard for deterministic embeddings into some spaces \cite{sidiropolous17,chawla2024}.

    \item Even if finding the optimal outlier set size for probabilistically embedding into HSTs is NP-hard, can the approximation factors in this paper be improved?

    \item Which other metric spaces admit perfect merge functions (and thus \bl extensions), and can we approximate optimal outlier embeddings in such spaces?
\end{itemize}


\bibliography{main,ArnoldBib}

\ifdagstuhl

\else 

\appendix 

\input{put_together_proofs}

\input{proofs}

\input{proofs_outlier}

\input{extra_applications}

\fi

\end{document}

%% file: lipicsheader.tex
    \input{macros}

    \title{Bi-Lipschitz extensions and outlier embeddings into trees}
    
    \author{Shuchi Chawla}{Department of Computer Science, University of Texas at Austin, USA  }{shuchi@cs.utexas.edu}{https://orcid.org/0000-0001-5583-2320}{This research was funded in part by NSF awards CCF-2217069 and CCF-2225259.}
    \author{Arnold Filtser}{Department of Computer Science, Bar-Ilan University, Ramat Gan, Israel}{arnold.filtser@biu.ac.il}{https://orcid.org/0000-0001-9578-9304}{This research was supported by the ISRAEL SCIENCE FOUNDATION (grant No. 1042/22).}
    \author{Kristin Sheridan}{Department of Computer Science, University of Texas at Austin, USA }{kristin@cs.utexas.edu}{https://orcid.org/0000-0002-5524-3259}{This research was funded in part by NSF awards CCF-2217069 and CCF-2225259 and by JPMorgan Chase \& Co.  Any views or opinions expressed herein are solely those of the authors listed, and may differ from the views and opinions expressed by JPMorgan Chase \& Co. or its affiliates. This material is not a product of the Research Department of J.P. Morgan Securities LLC. This material should not be construed as an individual recommendation for any particular client and is not intended as a recommendation of particular securities, financial instruments or strategies for a particular client.  This material does not constitute a solicitation or offer in any jurisdiction. }
    \author{Yonatan Trachtenberg}{Department of Computer Science, Bar-Ilan University, Ramat Gan, Israel}{ytrachtenberg18@gmail.com}{https://orcid.org/0009-0000-1181-3697}{}
    
    \authorrunning{S. Chawla, A. Filtser, K. Sheridan, Y. Trachtenberg} 
    \Copyright{Shuchi Chawla, Arnold Filtser, Kristin Sheridan, and Yonatan Trachtenberg} 
    
    \ccsdesc[500]{Theory of computation~Rounding techniques}
    \keywords{metric embeddings, hierarchically separated trees, outliers}
    \nolinenumbers 

    \category{APPROX}
    \relatedversion{}
    \relatedversiondetails{Full version}{https://arxiv.org/abs/2601.15470}
    
    \EventEditors{Mohit Singh and Tom Gur}
    \EventNoEds{2}
    \EventLongTitle{Approximation, Randomization, and Combinatorial Optimization. Algorithms and Techniques (APPROX/RANDOM 2026)}
    \EventShortTitle{\mbox{\scriptsize APPROX/RANDOM 2026}}
    \EventAcronym{APPROX/RANDOM}
    \EventYear{2026}
    \EventDate{August 19--21, 2026}
    \EventLocation{Boston University, Boston, Massachusetts, USA}
    \EventLogo{}
    \SeriesVolume{392}
    \ArticleNo{17}

%% file: macros.tex
\usepackage{enumerate}
\usepackage{tikz}
\ifdagstuhl \else
\usepackage[dvipsnames]{xcolor}
\fi
\usepackage{amsmath,amsfonts,amssymb,amsthm}
\usepackage{xspace}
\usepackage{cancel} 
\usepackage{mathtools}
\usepackage{braket,amsfonts}
\usepackage{thmtools}
\usepackage{algorithm}
\usepackage[noend]{algpseudocode}
\usepackage{lineno}

\newcommand{\hide}[1]{}

\ifdagstuhl

\renewcommand{\st}{^*}

\else 

\newtheorem{theorem}{Theorem}[section]

\newtheorem{lemma}[theorem]{Lemma}
\newtheorem{definition}[theorem]{Definition}
\newtheorem{observation}[theorem]{Observation}
\newtheorem{corollary}[theorem]{Corollary}
\newtheorem{claim}[theorem]{Claim}

\newcommand{\st}{^*}
\fi 

\newcommand{\reals}{\mathbb{R}}
\newcommand{\bl}{bi-Lipschitz\xspace}

\newcommand{\band}{band thickness\xspace}

\newcommand{\expec}{\mathop{\mathbb{E}}}
\newcommand{\cA}{\mathcal{A}}
\newcommand{\cD}{\mathcal{D}}

\newcommand{\cU}{\mathcal{U}}
\newcommand{\cH}{\mathcal{H}}
\newcommand{\cN}{\mathcal{N}}
\newcommand{\cX}{\mathcal{X}}
\newcommand{\cY}{\mathcal{Y}}
\newcommand{\cK}{\mathcal{K}}

\newcommand{\sminout}[1]{\mathcal{K}(#1)}
\newcommand{\sminoutval}[1]{\text{OutlierCost}(#1)}
\newcommand{\sminoutset}[1]{\mathcal{K}(#1)}
\newcommand{\probembed}[1]{\xhookrightarrow{#1}}

\newcommand{\embedcontr}[2]{\xhookrightarrow{#1,#2}}

\newcommand{\onion}{onion\xspace}
\newcommand{\Onion}{Onion\xspace}

\newcommand{\from}{\leftarrow}

\newcommand{\merge}{\texttt{MergeHST}}
\newcommand{\mergel}{\texttt{MergeL1}}

\newcommand{\goodmerge}{\text{perfect}}

\newcommand{\mergeplain}{\texttt{Merge}}
\newcommand{\pmerge}{\texttt{PerfectMerge}}
\newcommand{\dista}[1][\alpha]{d_{#1}}
\newcommand{\change}[1]{{\color{Mahogany} #1}}
\newcommand{\levelval}{\eta}
\newcommand{\eps}{\varepsilon}

\newtheorem{clm}{Claim}

\newenvironment{proofof}[1]{\vspace{0.1in}\noindent{\em Proof of #1.}}{\qed}

\usepackage[colorinlistoftodos,prependcaption,textsize=tiny]{todonotes}


%% file: arxivheader.tex
    \usepackage[margin=1in]{geometry}
    \usepackage[T1]{fontenc}
    \usepackage{lmodern}
    
    \usepackage{amsmath, amsthm, amssymb}
    \usepackage{hyperref}
    \usepackage{cleveref}
    \usepackage{thm-restate}
    
    \usepackage{authblk}

\input{macros}
    \title{Bi-Lipschitz extensions and outlier embeddings into trees}
    
    \author[1]{Shuchi Chawla\thanks{This research was funded in part by NSF awards CCF-2217069 and CCF-2225259.}}
    \author[2]{Arnold Filtser\thanks{This research was supported by the ISRAEL SCIENCE FOUNDATION (grant No. 1042/22).}}
    \author[1]{Kristin Sheridan\thanks{This research was funded in part by NSF awards CCF-2217069 and CCF-2225259 and by JPMorgan Chase \& Co.  Any views or opinions expressed herein are solely those of the authors listed, and may differ from the views and opinions expressed by JPMorgan Chase \& Co. or its affiliates. This material is not a product of the Research Department of J.P. Morgan Securities LLC. This material should not be construed as an individual recommendation for any particular client and is not intended as a recommendation of particular securities, financial instruments or strategies for a particular client.  This material does not constitute a solicitation or offer in any jurisdiction.}}
    \author[2]{Yonatan Trachtenberg}
    
    \affil[1]{Department of Computer Science, University of Texas at Austin, USA \protect\\ \texttt{\{shuchi, kristin\}@cs.utexas.edu}}
    \affil[2]{Department of Computer Science, Bar-Ilan University, Ramat Gan, Israel \protect\\ \texttt{arnold.filtser@biu.ac.il, ytrachtenberg18@gmail.com}}
    
    \date{}

%% file: intro.tex
\section{Introduction}

Low-distortion embeddings into trees are a fundamental algorithmic tool for reducing optimization problems on general graph metrics to the special case of tree metrics, where they are often easier to solve. The distortion of an embedding, defined as the maximum multiplicative change in pairwise distances, directly determines the loss in performance incurred by such a reduction.
A series of works \cite{alon1995graph,bartal1996,bartal1998,frt04} showed that 
{\em probabilistic} embeddings can embed any $n$-point metric into hierarchically separated trees (HSTs) with expected distortion $O(\log n)$.
While uniform distortion bounds are helpful in understanding the general effectiveness of using metric embeddings to solve a problem, from a practical perspective, we may ask how well a particular given metric $(X,d)$ can embed into trees. Indeed, for some metrics the optimal distortion is $\Omega(\log n)$ but for others it may be much smaller, e.g. $O(1)$, providing better algorithmic results.

{\bf In this paper we investigate instance-optimal distortion for probabilistic embeddings into trees by exploring the role of outliers.} For a finite metric $(X,d)$, a $(k,c)$-outlier embedding is a map defined on ``most of'' $X$---that is on all but $k$ vertices, which we designate as outliers---into another metric space $Y$ such that the distortion over the non-outliers is bounded by $c$. Intuitively, we can view the outliers as noisy or adversarial data that may hurt the quality of the instance-optimal embedding from $X$ into $Y$. The goal of the outlier embedding is to identify and remove these outliers upfront so as to obtain a much better embedding over the remaining points. Algorithmically, given a target distortion $c$, we want to find the smallest $k$ such that the removal of $k$ outliers enables embedding the rest of the metric with distortion at most $c$.
In this paper, we develop probabilistic outlier embeddings into trees, providing a tradeoff between outlier set size and distortion. Our main result is as follows.

\begin{restatable}[HSTs admit outlier embeddings]{theorem}{HSToutlier}\label{thm:outlier}
Let $(X,d)$ be any metric that admits a $(k,c)$ probabilistic outlier embedding into HSTs. Then for any constant $\epsilon>0$, there exists an algorithm that efficiently samples from a probabilistic embedding into HSTs of some $S\subseteq X$ with $|X\setminus S|\leq O(\frac{1}{\epsilon}k\log k)$ and expected distortion at most $(32+\epsilon)c$.
\end{restatable}


\paragraph*{Technical contribution: improved bi-Lipschitz extensions.} A key technical component in our construction is a bi-Lipschitz extension for embeddings of arbitrary finite metrics into trees. Given a finite metric $(X,d)$ and a low distortion embedding $\alpha$ from a subset $S\subset X$ into a host space $Y$, a bi-Lipschitz extension extends the domain of $\alpha$ to the entire set $X$, while only modestly increasing the distortion of the embedding. We show that if the distortion of the original embedding $\alpha$ is $c_{\alpha}$, it is possible to obtain a bi-Lipschitz extension with distortion at most $O(c_{\alpha}\log k)$ where $k=|X\setminus S|$. Formally, we achieve the following bounds, where the expansion factor $L_{exp}$ is the multiplicative increase in pairwise distances beyond the expansion of the original embedding $\alpha$, and the contraction factor $L_{contr}$ is the multiplicative decrease beyond the contraction in $\alpha$.



\begin{restatable}[HSTs have probabilistic bi-Lipschitz extensions]{theorem}{HSTblipthm}\label{thm:bilip-hst}
Let $\cH$ be the class of all HSTs.
For any subset $S\subseteq X$ with $k:=|X\setminus S|$ and a probabilistic embedding $\cD_S:S\probembed{}\cH$, there exists a probabilistic bi-Lipschitz extension of $\cD_S$, $\cD:X\probembed{}{}\cH$, with expansion factor $L_{exp}=O(\log k)$ and contraction factor $L_{contr}=4$.

Further, we can efficiently sample from the distribution $\cD$ if we can efficiently sample from $\cD_S$. 

\end{restatable}

A connection between outlier embeddings and Lipschitz extensions was first established by Chawla and Sheridan \cite{chawla2024} in the context of deterministic embeddings into normed metrics. They showed that deterministic embeddings into $\ell_2$ admit a Lipschitz extension with expansion factor $O(\log k)$ (but potentially unbounded contraction), and this suffices to obtain a bicriteria approximation for outlier embeddings into $\ell_2$. We extend this approach to embeddings into HSTs by combining it with an LP-rounding approximation algorithm of 
Munagala et al. \cite{mungala2023} for instance-optimal probabilistic embedding into HSTs.

\paragraph*{Probabilistic extension and bounded contraction.}
Extending the approach of \cite{chawla2024} to HSTs poses two new technical challenges. First, while Chawla and Sheridan \cite{chawla2024}'s construction can tolerate arbitrary contraction in the Lipschitz extension, in order to employ Munagala et al.'s LP relaxation for HST embeddings, we necessarily need to control contraction as well.
The second challenge is applying the notion of Lipschitz extension to probabilistic embeddings. This is important because probabilistic tree embeddings can obtain exponentially smaller distortion than deterministic tree embeddings. Consider, in particular, a distribution $\cD$ over embeddings from some subset $S$ of a given metric $(X,d)$ into a host space $Y$. Simply drawing an embedding $\alpha$ from $\cD$ and extending it to $X$ may result in distortion $\propto\mathbb{E}_{\alpha\sim \cD}[\max_{u,v\in X}(d_{\alpha}(u,v)/ d(u,v))]$. This can be much larger than the desired distortion $\max_{u,v\in X} (\mathbb{E}_\alpha[d_{\alpha}(u,v)]/ d(u,v))$.

We resolve both challenges by constructing a bi-Lipschitz extension through a sequence of ``merge'' operations, each of which combines two embeddings over domains that intersect in a single point. By carefully aligning and stitching each such pair of embeddings, we ensure strong bounds on both contraction and expansion.

\paragraph*{Improved expansion via ``onion-like'' partitions}
Chawla and Sheridan's Lipschitz extension for $\ell_2$ employs a probabilistic partitioning procedure over the outliers $X\setminus S$ similar to that of Calinescu, Karloff, and Rabani \cite{calinescu2005extension}, and attaches an embedding of each component of the partition to its closest point in $S$ (that we call the {\em anchor point}). This approach combined with the merge operation described above yields a {\em \bl} extension for HSTs with expansion factor $O(\log^2 k)$. We develop a new technique to obtain a quadratic improvement over this factor, leading to \Cref{thm:bilip-hst}. Our key observation is that the largest distortion is suffered by pairs that are separated into different components by the CKR partitioning procedure---one logarithmic factor arises from the probability of separation, and the second arises from the distortion between the points and their respective anchor points. 

To improve upon this distortion, we construct a probabilistic partitioning where each component has an ``onion-shell'' shape---all points are roughly equidistant from the anchor point. We then employ an embedding of each component that preserves distances to the anchor point within a constant factor (while potentially inflating other distances by a logarithmic factor). This allows us to bound the distortion between any pair of points by $O(\log k)$ times the distortion of the embedding over $S$.

Our improved partitioning approach also applies to $\ell_1$, providing a \bl extension with expansion factor $O(\log k)$. \cite{chawla2024} previously obtained a \bl extension for $\ell_1$ with expansion factor $O(\frac 1 {c_S} \log^2 k+\log k)$, where $c_S$ is the distortion of the embedding of $S$ into $\ell_1$; this can be a logarithmic factor worse than the bound we achieve when $c_S$ is $O(1)$.

\begin{theorem}[$\ell_1$ has \bl extensions, informal]\label{thm:bilipl1}
Deterministic $\ell_1$ embeddings have deterministic \bl extensions with expansion factor $O(\log k)$ and contraction factor $O(1)$.
\end{theorem}

\subsection{Related work}
Embedding into HSTs has been studied primarily with the objective of minimizing expected distortion. See e.g. \cite{alon1995graph,bartal1996,bartal1998,charikar1998,frt04,Bartal04,KGR25}. 
In our work, we will use as a subroutine the algorithm of Fakcharoenphol et al. \cite{frt04}, which for any $n$ point metric samples from an $O(\log n)$ distortion probabilistic embedding into HSTs.
{\em Instance-optimal} distortion has been studied for target spaces such as constant dimensional Euclidean space \cite{badoiu2006,matousek2008,deberg2010,edmonds2010,sidiropulos2019}, the line  \cite{badoiu2005,matousek2008,fellows2013,nayyeri2015}, trees \cite{badoiu2007,chepoi2012}, and ultrametrics \cite{alon2008,mungala2023}.
The goal in such embeddings is to minimize the distortion needed for the given {\em specific} $(X,d)$, which in some cases may be smaller than the best bound achievable for all metrics of that size.


The notion of outlier embeddings was initially defined by Sidiropolous et al. \cite{sidiropolous17} when they studied deterministic embeddings into ultrametrics, trees, and constant-dimensional $\ell_2$ space. They consider approximating {\em deterministic, isometric} outlier embeddings (where the distortion $c$ is 1) and bicriteria approximations with additive distortion. The structural properties used for isometric approximations don't apply to non-isometric embeddings, and techniques
for bounding additive distortion tend to be quite different from those for multiplicative distortion, the focus of this work. 
Chubarian and Sidiropolous \cite{Chubarian20} continued the work on outlier embeddings and studied non-isometric outlier embeddings of unweighted graph metrics into the line. Later, Chawla and Sheridan \cite{chawla2024} studied deterministic outlier embeddings into $\ell_2$ space with unbounded dimension. 

Several related notions of distortion were previously studied. 
One notion is {{\em embeddings with slack}, in which the goal is to minimize the number of pairs that are distorted by more than some target factor \cite{abraham2005,chan2006,lammersen2009}. 
An even stronger notion is that of \emph{scaling distortion}, where for every $\eps>0$, at most $\eps$ fraction of the pairs suffer from distortion $\ge f(\eps)$ (for some function $f$), see \cite{KSW04,ABN11,ABN15,BFN19}. 
Another notion is \emph{terminal distortion}, where instead of an outlier set, there is a subset of terminals $K\subseteq X$, and the goal is only to preserve distance pairs involving a terminal $K\times X$ \cite{EFN17,MMMR18,NN19,CN24}.
A  stronger notion is \emph{prioritized distortion} \cite{EFN18,EN22,FGN24} where there is a priority ordering over the vertices, and the distortion guarantee is w.r.t. the priorities. 
Another notion is that of local embeddings
\cite{abraham2009,arora2012,charikar2010} where the goal is to minimize the distortion from each point to its $k$ nearest neighbors.

Outlier embeddings are also closely related to metric Ramsey theory \cite{BFM86,bartal2003,BLMN05,BLMN05b,mendel2007ramsey,NT12,BGS16,ACEFN20,Bar21,FL21,FL22tw}. Given inputs $n$ and $c$, metric Ramsey theory asks: What is the largest $m$ such that {\em all} metrics of size $n$ have a submetric of size $m$ that embeds into some target space with distortion at most $c$? In this paper we are primarily focused on {\em instance-optimal} solutions rather than universal bounds that apply to all metrics, and we seek to approximate the size of the outlier set rather than the size of the non-outlier set. 

Many other types of embedding compositions or extensions have been considered. One particular example 
is the notion of a Lipschitz extension \cite{Johnson1984,LN05,naor2017}, which seeks to leave non-outliers unaltered from the original embedding and ensure the expansion (but not necessarily the contraction) of every pair of nodes is bounded. Given an embedding on $S\subseteq X$, Calinescu, Karloff, and Rabani \cite{calinescu2005extension} develop a partitioning algorithm for $X\setminus S$ that implicitly defines Lipschitz extensions with a contraction factor at most $O(\log |S|)$ for all metric spaces. Chawla and Sheridan \cite{chawla2024} define essentially the same partitioning algorithm, except that in their algorithm, \cite{calinescu2005extension} pick a uniformly random permutation over the set $S$, but \cite{chawla2024} take a random permutation over the set of $X\setminus S$'s closest neighbors in $S$. The resulting Lipschitz extension factor is $O(\log |X\setminus S|)$. We will permute over the same set as \cite{chawla2024}, but we follow the cleaner analysis of \cite{calinescu2005extension}.
Another related notion is that of $0$-extension \cite{calinescu2005extension,FHRT03,EGKRTT14,FKT19} where one seeks a stochastic contraction $f:X\rightarrow K$ from the metric point set to a subset $K$ of terminals while minimizing $\max_{x,y\in X}\mathbb{E}\left[\frac{d_X(f(x),f(y))}{d_X(x,y)}\right]$.


%% file: definitions.tex
\section{Definitions and results \label{sec:definitions}}

\subsubsection*{Embeddings and distortion}

A metric space is a set of points $X$ with distances $d_X(\cdot,\cdot)$ between them that obey the triangle inequality. That is, for all $x,y,z\in X$, $d_X(x,z)\leq d_X(x,y)+d_X(y,z)$. 
{}{We assume without loss of generality (via scaling) that in the metrics we consider, the minimum distance between distinct pairs of points is at least $1$. For subsets $S,T\subseteq X$, denote $d(S,T):=\min_{s\in S,t\in T}\set{d(s,t)}$
}
If $\alpha:X\rightarrow Y$ is a map from $(X,d_X)$ to $(Y,d_Y)$, we write $\dista(x,y) := d_Y(\alpha(x),\alpha(y))$ to denote the distance between the images of two points $x, y \in X$ under this map. 


$\cD$ is a probabilistic embedding of $(X,d_X)$ into $\cY$ if for each function $\alpha$ in the support of $\cD$, $\alpha:X\rightarrow Y(\alpha)$ for some $(Y(\alpha),d_{Y(\alpha)})\in \cY$. This embedding has expected {\em expansion} $c_{exp}\geq 1$, {\em contraction} $c_{contr}\geq 1$, and expected {\em distortion} $c_{contr} \cdot c_{exp}$ if for all $x,y\in X$ we have: 
    \begin{align*}
    \dista(x,y) & \geq \frac{1}{c_{contr}} \cdot d(x,y) \quad\quad \forall 
    \alpha \in \operatorname{support}(\cD) \\
        \expec_{\alpha\sim \cD}[\dista(x,y)] & \leq c_{exp} \cdot d(x,y).
    \end{align*}
Observe that {\em every} embedding in the support of $\cD$ has contraction at most $c_{contr}$, but pairwise expansion (and thus distortion) is bounded in expectation over the randomness in $\cD$. If the support of $\cD$ contains a single embedding, it is a {\em deterministic} embedding. Further, an embedding is {\em expanding} or {\em non-contracting}
if it has contraction $1$. 

We use the notation $\alpha: X \embedcontr{c_{exp}}{c_{contr}} Y$ to denote a deterministic or probabilistic embedding with expansion $c_{exp}$ and contraction $c_{contr}$ from $(X,d_X)$ into $(Y,d_Y)$.
If $c_{contr}$ is not specified, the embedding is expanding.

\subsubsection*{Bi-Lipschitz extensions} 


\begin{definition}[Bi-Lipschitz extension]\label{def:bilip}
Let $(X,d)$ be a metric, $S\subseteq X$, and $\alpha_S:S\embedcontr{c_{S}}{c_{S}'}Y$ for some parameters $c_{S}',c_S$. Then an embedding $\alpha:X\rightarrow Y$ is a (strong) {\em bi-Lipschitz extension} of $\alpha_S$ with extension factors $L_{exp},L_{contr}$ if \footnote{Note that $L_{exp},L_{contr}$ could be factors that depend on $X$ and $S$.}
    \begin{align}
        & \text{for all } u\in S, &  \alpha(u)&=\alpha_S(u), \label{eq:preserveembed} \\
        & \text{for all } u,v\in X& 
        \dista(u,v) & \geq \frac{1}{c_S' \cdot L_{contr}}d(u,v)\\
        & \text{and, for all } u,v\in X, & \dista(u,v) & \leq L_{exp}\cdot c_S \cdot d(u,v). \label{eq:weaknested-constraint2}
    \end{align}
If $Y$ is a normed vector space, for Constraint \ref{eq:preserveembed} we require only that the restriction of $\alpha$ to the subspace spanned by $S$ matches the embedding 
$\alpha_S$.\footnote{
Notably, to get non-trivial extensions we must allow $\alpha$ to ``extend'' the dimension of $\alpha_S$ in some cases. To see this, consider an isometric (distortion $1$) embedding of one or two points of a length $n$ cycle into the real line. Any \bl extension that also embeds into the real line must have $L_{ext}\cdot L_{contr}=\Omega(n)$, as $\Omega(n)$ distortion is required to embed the cycle into a line. }

Further, $\alpha$ is a {\em weak bi-Lipschitz extension} with extension factors $L_{exp},L_{contr}$ if Constraint \ref{eq:preserveembed} is replaced with the following weaker constraint:
\begin{align*}
    & \text{for all } u,v\in S, &  d_\alpha(u,v)&=d_{\alpha_S}(u,v).
\end{align*}
\end{definition}


In this paper, we discuss  bi-Lipschitz extensions in the probabilistic setting. In particular, we model a probabilistic bi-Lipschitz extension on the weak version of a deterministic bi-Lipschitz extension, requiring that the {\em joint distribution} of distances between pairs of nodes in $s$ is the same under the extended embedding.

\begin{definition}[Probabilistic bi-Lipschitz extension]\label{def:prob-bilip}
Let $(X,d)$ be a metric, $S\subseteq X$, and $\cD_S\embedcontr{c_S}{c_S'}\cY$ for some parameters $c_S,c_S'$. Then a probabilistic embedding $\cD:X\rightarrow \cY$ is a {\em probabilistic bi-Lipschitz extension} of $\alpha_S$ with extension factors $L_{exp},L_{contr}$ if
\begin{align}
&(d_\alpha(x,y))_{x,y\in S\times S} \text{ has the same joint distribution for }\alpha \sim \cD \text{ and }\alpha\sim \cD_S \\
&\text{for all } x,y\in X, \text{ and for all } \alpha\in\operatorname{support}(\cD), \quad \quad 
        \dista(x,y)  \geq \frac{1}{L_{contr}\cdot c_S'}\cdot d(x,y) \\
&\text{and, for all } x,y\in X, \quad \qquad\ \ \ \ \    \ \ \ \ \ \ \ \ \ \ \ \ \ \ \expec_{\alpha\sim \cD}[\dista(x,y)]  \leq L_{exp}\cdot c_S \cdot d(x,y). 
    \end{align}

\end{definition} 


%

\subsubsection*{Outlier embeddings}

If $\cD$ is a probabilistic embedding from $(X,d_X)$ into $\cY$ and $K\subseteq X$, define $\cD_{X\setminus K}$ to be the distribution identical to $\cD$ but  with all functions $\alpha$ in the support restricted to the set $X\setminus K$ (i.e. we can sample from $\cD_{X\setminus K}$ by sampling from $\cD$ and restricting the domain of the function drawn).

\begin{definition}\label{randomoutlier}
    A probabilistic embedding $\cD$ of  $(X,d_X)$ into $\cY$ is a {\em probabilistic $(k,c)$-outlier embedding} if there exists $K\subseteq X$ such that $|K|\leq k$ and $\cD_{X\setminus K}$ is a probabilistic embedding of $(X\setminus K,d_X)$ into $\cY$ with distortion $c$. 
\end{definition}

Note that in Definition \ref{randomoutlier}, the  subset $K$ is chosen {\em deterministically} (not sampled). The distances from $K$ to other nodes (and each other) are ignored. The randomness is only on the embedding for the remaining nodes in $X\setminus K$. 

\subsubsection*{Hierarchically Separated Trees (HSTs), ultrametrics, and $\ell_p$}


\begin{definition}
    A metric $(X,d)$ is a {\em $\beta$-hierarchically separated tree ($\beta$-HST)} if there is a map from $X$ to the leaves of a tree such that the following hold:
    \begin{itemize}
        \item Each node $u$ in the tree is associated with (``labeled'' with) some value $\levelval_u$, and $\levelval_v=\beta\cdot \levelval_u$ whenever $u$ is a child of $v$ and $u$ is not a leaf.\footnote{In this paper, we will use the term $\beta$-HSTs to refer to what other sources sometimes call exact $\beta$-HSTs. These sources typically define non-exact $\beta$-HSTs as only requiring $\levelval_v\geq \beta\cdot \levelval_u$.  }   Further, for all leaves $v$, $\levelval_v = 0$, and for nodes $v$ that are parents of leaves, $\levelval(v)=1$. 
        \item For two leaves  $x,y\in X$ with least common ancestor $z$, $d(x,y)=\levelval_z$. 
    \end{itemize}
\end{definition}

For simplicity, we will primarily focus on $2$-HSTs, and in referring to HSTs, we mean $2$-HSTs unless otherwise noted. We use $\cH$ to denote the set of all $2$-HSTs. Note that focusing on $\beta$-HSTs for a different constant $\beta$ at worst changes results by a constant factor. 

\begin{definition}\label{def:ultrametrics}
An {\em ultrametric} is a metric $(X,d)$ that fulfills a strong triangle inequality guarantee: for all $x,y,z\in X$, $d(x,y)\leq \max\set{d(x,z),d(z,y)}$. 
\end{definition}

HSTs form a subset of the set of {\em ultrametrics}. Further, any ultrametric embeds into an HST with distortion at most $2$ \cite{bartal2003}. We will use $\cU$ to denote the set of all ultrametrics.

\begin{definition}[$\ell_p$ space]
    $\ell_p$ in dimension $\delta$, denoted $\ell_p^\delta$, is $(\reals^\delta,d_{\ell_p})$ where for $x,y\in \reals^\delta$, $d_{\ell_p}$ is defined as
    \begin{align*}
        d_{\ell_p}(x,y):= \left(\sum_{i=1}^\delta |x_i-y_i|^p\right)^{1/p}.
    \end{align*}
We let $\ell_p$ space refer to the set of all spaces $\ell_p^\delta$ for $\delta\in[1,\infty)$.
\end{definition}

\subsubsection*{Main results}


Our primary results are for probabilistic bi-Lipschitz extensions of embeddings into HSTs and deterministic bi-Lipschitz extensions of embeddings into $\ell_1$ space. 

\HSTblipthm*




We can also argue for the existence of deterministic extensions for $\ell_1$ embeddings, though we do not know how to find such deterministic embeddings efficiently.

\begingroup
\renewcommand\thetheorem{\ref{thm:bilipl1}}
\begin{restatable}[$\ell_1$ has \bl extensions]{theorem}{biliplonethm}
For any subset $S\subseteq X$ with an embedding $\alpha_S:S\embedcontr{c_S}{c_S'}\ell_1$, there exists
\begin{enumerate}
    \item a deterministic {\em weak} bi-Lipschitz extension of $\alpha_S$, $\alpha:X\probembed{}\ell_1$ with extension factors $L_{exp}=O(\log k)$ and $L_{contr}=1$, where $k:=|X\setminus S|$ and

    \item a deterministic strong bi-Lipschitz extension with $L_{exp}=O(\log k)$ and $L_{contr}=7$.
\end{enumerate}
\end{restatable}
\addtocounter{theorem}{-1} 
\endgroup





Our \bl extensions into $\ell_1$ are asymptotically optimal. Some metrics on $n$ points (namely expander graph metrics) require $\Omega(\log n)$ distortion for embedding into $\ell_1$ \cite{matousek1997lplowerbounds}. We can always ``extend'' a distortion $1$ embedding on one or two points to an embedding on an expander graph metric, obtaining an embedding of distortion at most $L_{exp}\cdot L_{contr}$. Thus, we must have $L_{exp}\cdot L_{contr}=\Omega(\log k)$.
Likewise, some metrics require $\Omega(\log n)$ distortion to embed  probabilistically into HSTs \cite{bartal1996}, so our HST results are also asymptotically optimal.



Our results for probabilistic \bl extensions into HSTs allow us to obtain a bicriteria approximation for probabilistic outlier embeddings into HSTs, which we discussed in the introduction. 

\HSToutlier*


In Section \ref{sec:hst}, we expand our results to the weighted outlier setting, and we use  a modification of the LP from Munagala et al. \cite{mungala2023} to prove Theorem \ref{thm:outlier}.
Our LP-based techniques also allow us to obtain the following result for the weighted variant of the problem.\footnote{In the weighted outlier set problem we are given weights $w_i$ on all nodes $i$ in the metric space. The cost of a subset $K$ of the metric space is $w(K):=\sum_{i\in K}w_i$, and the goal is to find an outlier embedding with some target distortion such  that the cost of the outlier set is minimized.}

\begin{corollary}\label{thm:weighted}
Given a metric $(X,d)$, node weights $w$, and target distortion $c$, let $w\st$ be the optimal outlier set cost for probabilistically embedding $(X,d)$ into HSTs with $c$ distortion and let $k$ be the size of such an optimal outlier set. Then for any constant $\epsilon>0$, there exists an algorithm that efficiently samples from a probabilistic embedding into HSTs of some $S\subseteq X$ with $w(X\setminus S)\leq O(\frac{1}{\epsilon}w\st\log k)$ and expected distortion at most $(32+\epsilon)c$. 
\end{corollary}

%% file: components.tex
\section{Components of a bi-Lipschitz extension}\label{sec:components}

The  bi-Lipschitz extensions we discuss in this paper have three primary components. In particular, given a probabilistic embedding $\cD_S:S\rightarrow \cY$ for $S\subseteq X$, we develop bi-Lipschitz extensions for $\cD_S$ when the relevant spaces have certain properties. Namely, the space $X$ should admit what we refer to as ``\onion partitions'' and the space $\cY$ should admit ``bounded-diameter embeddings'' and ``perfect merge functions.'' We show in  Section \ref{sec:partitions} that \onion partitions always exist, so our bi-Lipschitz extensions depend on the existence of bounded-diameter embeddings and perfect merge functions for the space into which we are embedding. 

At a high level, we will use the \onion partitions to find clusters of the set $K=X\setminus S$ and assign an ``anchor point'' in $S$ to each such cluster. Then we find a bounded-diameter embedding on each cluster with its anchor point, and we ``merge'' each of these embeddings with  the original embedding on $S$ using a perfect merge function.

\subsection{\Onion\ partitions}

We begin by defining what we call \onion partitions. Given a metric $(X,d)$ and a subset $S\subseteq X$ with $K=X\setminus S$, an \onion partition of $K$ is a partition $\cK=\set{K_i}$ of $K$ such that each $K_i$ has an assigned ``anchor point'' from $S$. For any point $x \in K_i$, we refer to its cluster's assigned point in $S$ as the anchor point of $x$. The partition should satisfy three properties: (1) for any $x,y\in K$, the probability that $x$ and $y$ are separated in the partition is proportional to $d(x,y)/d(\set{x,y},S)$; (2) for each $x\in K$, the distance between $x$ and its anchor point is (deterministically) similar to $d(x,S)$; and (3) for any cluster $K_i\in \cK$, the distance between $x\in K_i$ and its assigned anchor point is (deterministically) not too much smaller than the diameter of the cluster $K_i$ itself.

\begin{definition}[\Onion\ partitions]\label{def:tentpartition}
Let $(X,d)$ be a metric and let $S\subseteq X$ with $K:=X\setminus S$. Then an {\em \onion partition} of $K$ is a joint distribution over $(\cK=\set{K_i}_{i\in \mathbb{N}},r:\cK\rightarrow S)$ such that $\cK$ is a partition of $K$. 

For any $(\cK,r)$ in the support of the distribution and $K_i\in \cK$, we call $r(K_i)$ the {\em anchor point} of $K_i$. We extend the domain of $r$ to  points in $K$: for $x\in K$, $r(x)$ denotes the anchor point assigned to the component containing $x$ in the random partition $(\cK,r)$. Likewise, let $\Delta_{K_i}$ denote the diameter of a component $K_i$, and $\Delta_x$ denote the (random) diameter of the component containing $x$ in the random partition $(\cK,r)$. Additionally, for any fixed $(\cK,r)$ we let $K_i^+$ denote the set $K_i\cup \set{r(K_i)}$ for each $K_i\in \cK$.




An onion partition has {\em separation parameter} $\rho>0$, {\em closeness} $b>1$, and {\em \band} $\gamma>1$ if it meets the following three criteria:

\begin{enumerate}
    \item \textbf{Low separation probability.} For any $x,y\in K$, the probability that $x$ and $y$ are in different components of the partition is upper bounded by $\rho\cdot \frac{d(x,y)}{d(\set{x,y},S)}$
    
    \item \textbf{An anchor point is a ``close'' point in $S$.} 
    For any $x\in K$ and any $(\cK,r)$ in the support of the distribution, $d(x,r(x))\leq b\cdot d(x,S)$ with probability $1$.

    \item \textbf{Cluster points are arranged in a thin band around their anchor point.} For any $x\in K$ and any $(\cK,r)$ in the support of the distribution, $d(x,r(x))\geq \Delta_x/\gamma$ with probability $1$.
\end{enumerate}


    
\end{definition}

\begin{figure}
    \centering
    \includegraphics[width=0.5\linewidth]{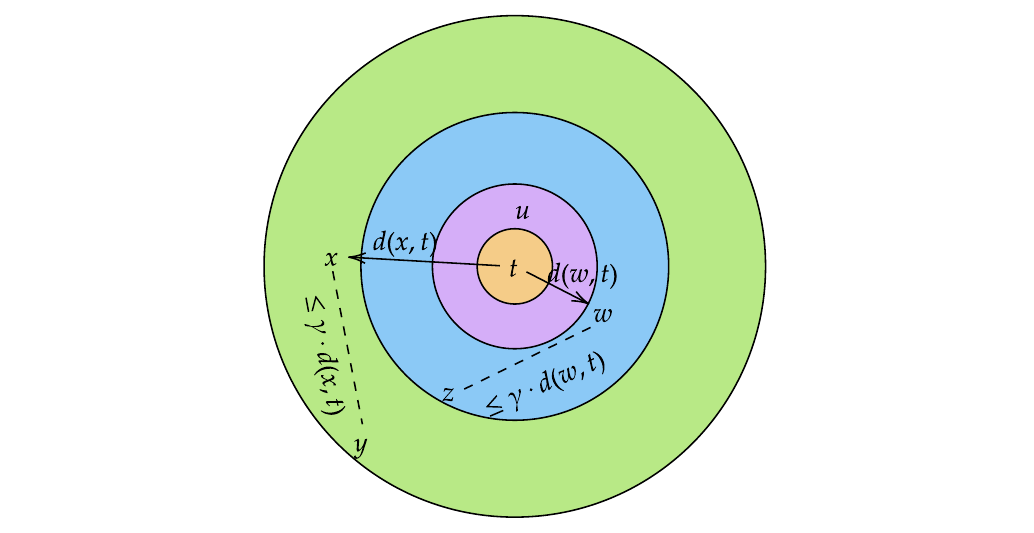}
    \caption{This figure depicts the third property of \onion partitions. In particular, for any point $t\in S$, we think of the clusters with anchor point $t$ as rings around $t$ (like the rings of an onion). For any $x,y\in K$, if $x$ and $y$ are in the same cluster with anchor point $r(x)=r(y)=t$, then the third criterion of \onion partitions states that they must be a distance at most $\gamma \cdot d(x,t)$ apart. }
    \label{fig:onion}
\end{figure}

\Cref{fig:onion} gives a visual depiction of the third property of \onion partitions, and in \Cref{sec:partitions}, we show the following theorem about their existence.

\begin{restatable}[\Onion partitions]{theorem}{Onionthm}\label{thm:onion}
For any metric $(X,d)$ and any $S\subseteq X$, there exists an \onion\ partition on $K=X\setminus S$ with separation parameter $O(\log k)$, closeness $2$, and \band $6$, where $k=|K|$.
\end{restatable}



\subsection{Bounded-diameter embeddings}

The second component of our extension method involves what we call bounded-diameter embeddings. In particular, we consider expanding embeddings that are parameterized not only by their worst expected expansion for any pair of nodes, but also by their expected expansion of the {\em diameter} of the original metric. That is, a bounded-diameter embedding can place a smaller upper bound on the expected expansion for pairs of nodes that are relatively far apart in the original metric than it does for other pairs of nodes. We define this formally in \Cref{def:bound-diam}.

\begin{definition}[Bounded diameter embedding]\label{def:bound-diam}
A probabilistic embedding $\cD:X\rightarrow \cY$ is an expanding {\em bounded-diameter embedding} with expansion $c$ and {\em diameter expansion} $c_\Delta$ if the following conditions hold:
\begin{align*}
    \dista(x,y) & \geq d(x,y) \quad\quad &\forall x,y\in X, \forall 
    \alpha \in \operatorname{support}(\cD) \\
        \expec_{\alpha\sim \cD}[\dista(x,y)] & \leq c \cdot d(x,y)  \quad \quad &\forall x,y\in X \\
    \expec_{\alpha\sim \cD}[d_\alpha(x,y)]
    &\leq c_\Delta \cdot \max_{x',y'\in X}\set{d(x',y')} \quad \quad &\forall x,y\in X
\end{align*}

Further, we say {\em $\cY$ admits bounded-diameter embeddings with expansion $c_k$ and diameter expansion $c_\Delta$} if for all metrics $(X,d)$ of size $k$, there is a bounded-diameter embedding of $X$ into $\cY$ with the given parameters.
\end{definition}

Notably, if all metrics of size $k$  can embed into $\cY$ with expected expansion at most $c_k$, then $\cY$ admits bounded-diameter embeddings with expansion and diameter expansion equal to $c_k$. However, in the settings we consider here (namely HSTs and $\ell_1$ space), we will consider bounded-diameter embeddings where $c_\Delta$ is tighter than $c_k$. In this paper, we will use FRT embeddings \cite{frt04} as our bounded diameter embeddings.

\begin{theorem}[\cite{frt04}]\label{thm:frt}
For any metric $(X,d)$ with $|X|=n$, there exists a probabilistic expanding bounded diameter embedding (an FRT embedding) of $X$ into HSTs with expansion $O(\log k)$ and diameter expansion $2$. 
\end{theorem}

In fact, for HST embeddings, we can assume without loss of generality an even stronger property than that required by \Cref{def:bound-diam}: For any embedding $\alpha$ in the support of an HST embedding on $X$, the diameter of the image of $\alpha$ is at most twice the diameter of $X$.  
To see this, let $\Delta$ be the diameter of $X$ and consider an HST embedding with diameter larger than $2^{\lceil \log \Delta \rceil}$. The nodes of the HST at level $2^{\lceil \log \Delta \rceil}$ or higher can all be merged into a single node without incurring any new contraction between any pair of points, as their new distance under this contraction is the minimum of their old distance and $2^{\lceil \log \Delta \rceil}$.

Further, trees embed isometrically into $\ell_1$, so we immediately get the following corollary. 

\begin{corollary}\label{cor:l1diamembed}
For any metric $(X,d)$ with $|X|=n$, there exists a probabilistic expanding bounded diameter embedding (an FRT embedding) of $X$ into $\ell_1$ with expansion $O(\log k)$ and diameter expansion $2$. 
\end{corollary}

\subsection{Perfect merge functions}

In the final step of our extension method, we ``merge'' the embedding $\alpha_S$ on $S\subseteq X$ with bounded-diameter embeddings on the clusters of our \onion partition. 
In particular, a {\em merge function}  is a function whose input is two embeddings whose domains overlap at a single point and whose output is a  new embedding on the union of their domains. For our purposes, we want to ensure some of the necessary distance properties are maintained. In particular, for any pair of nodes that appeared in the domain of the same input embedding, the distance under the new embedding should be the same as under the original one. Additionally, we want to ensure that no pair of nodes is allowed to be ``too close'' in the new embedding. This extra requirement will allow us to bound contraction in our final bi-Lipschitz extension.



\begin{definition}[Perfect merge]\label{def:perfectmerge}
A procedure $\mergeplain$  is a {\em \goodmerge\ merge function} for embedding $X$ into $\cY$ with {\em contraction factor} $\eta$ if for any $Z_1,Z_2\subseteq X$ with $Z_1\cap Z_2 = \{v\}$ and  any deterministic embeddings $\alpha_1:Z_1\rightarrow Y_1$, $\alpha_2:Z_2\rightarrow Y_2$ for $Y_1,Y_2\in \cY$, it produces an embedding $\alpha:Z_1\cup Z_2\rightarrow Y$ for some $Y\in \cY$ such that:
\begin{enumerate}
    \item For all $x,y\in Z_1$, $\dista(x, y)=\dista[\alpha_1](x,y)$.
    \item For all $x,y\in Z_2$, $\dista(x, y)=\dista[\alpha_2](x,y)$.
    \item For $x\in Z_1$ and $y\in Z_2$, we have $d_\alpha(x,y)\geq \frac{d_{\alpha_1}(x,v)+d_{\alpha_2}(v,y)}{\eta}$.

\end{enumerate}
\end{definition}

\noindent We say $\cY$ admits perfect merge functions if there is a perfect merge function for any metric $X$.

In Section \ref{sec:merge}, we prove the following existence results about perfect merge functions for HSTs and $\ell_1$.

\begin{lemma}\label{lem:merge-hst}
HSTs admit a perfect merge function with contraction factor $2$. 
\end{lemma}

\begin{lemma}\label{lem:merge-l1}
$\ell_1$ admits a perfect merge function with contraction factor $1$.
\end{lemma}

\subsection*{Putting the components together}

With our three components defined, we now put them together to construct \bl extensions for amenable target spaces $\cY$. 
In particular, given a metric $(X,d)$ with subset $S\subseteq X$ and $\cD_S:S\rightarrow \cY$, our extension proceeds in the following manner:

\begin{enumerate}
    \item Sample an embedding on $\alpha_S$ from $\cD_S$.
    \item Sample an onion partition $(\cK,r)$ of $K=X\setminus S$.
    \item For each $K_i\in \cK$, sample an embedding $\alpha_i$ from a bounded diameter embedding of $K_i^+=K_i\cup \set{r(K_i)}$.
    \item Use a perfect merge function to sequentially merge each $\alpha_i$ with $\alpha_S$.
\end{enumerate}

\Cref{alg:extend} formalizes this sampling process, and \Cref{thm:merge-components} states that the embedding sampled by \Cref{alg:extend} samples from a \bl extension of the input embedding $\cD_S$. We highlight in red the places where \Cref{alg:extend} uses each of the three components we have discussed so far.


\input{alg_extend}


\begin{theorem}\label{thm:merge-components}
Let $(X,d)$ be a metric with $S\subseteq X$ and let $\cD_S:S\rightarrow \cY$ be a probabilistic embedding with contraction at most $c_S'$ and expansion at most $c_S$. 

\noindent Then, if:
\begin{enumerate}
    \item $X$ admits an \onion partition with separation parameter $\rho$, closeness $2$,\footnote{This could be  defined for more general closeness parameters, but it greatly complicates the algebra, so we set $b=2$ in our analysis.} and \band $\gamma$,
    \item $\cY$ admits expanding bounded-diameter embeddings on sets of size $k$ with expansion $c_k$ and diameter expansion $c_\Delta$, and  
    \item $\cY$ admits a perfect merge function with contraction factor $\eta$,
\end{enumerate}
\noindent $\cD_S$ has a bi-Lipschitz extension to all of $X$ with extension factors $L_{exp}=O(\gamma \cdot \rho\cdot (c_\Delta/c_S+1)  +c_k/c_S)$ and $L_{contr}=\eta^2$.
Further, \Cref{alg:extend} samples from such an extension.
\end{theorem}


Before proving this theorem, we observe that due to the first two properties of a perfect merge function, it is sufficient to bound the first extension \Cref{alg:extend} produces that includes both $x$ and $y$ in the domain.

\begin{observation}\label{fact:induct}
Let $x,y\in X$. Consider the first iteration of \Cref{alg:extend} in which both $x$ and $y$ are in the domain of $\alpha$, and let $\hat{d}$ be $\dista(x,y)$ at that time step. Then at every subsequent iteration of the algorithm, $\dista(x,y)=\hat{d}$.
\end{observation}

We are now ready to bound the contraction and expansion of the embedding procedure, which we do in \Cref{lem:contraction-general,lem:expansion-general}. We defer the full proofs of these lemmas to
\ifdagstuhl
the full version of this paper
\else 
\Cref{sec:exp_contr_proofs}
\fi
and just give a brief overview of the ideas here.

\begin{lemma}[Contraction bound]\label{lem:contraction-general}
Let $X,S,\cD_S$ be as in \Cref{thm:merge-components} and let $\alpha:X\rightarrow \cY$ be the embedding returned by \Cref{alg:extend}. Then for all $x,y\in X$, $d_\alpha(x,y)\geq d(x,y)/\eta^2$. 
\end{lemma}


To intuitively understand the contraction bound, note that when some node $v\in K_i$ is added to the domain of $\alpha$, there is no contraction between $v$ and other nodes in $K_i^+$, as $\cD_i$ is non-contracting and we use a perfect merge function to combine this with $\alpha$. Further, for any node $u$ already in the domain of $\alpha$ when $v$ was added, property 3 of perfect merge functions ensures $d_\alpha(u,v)\geq ({d_\alpha(u,r(v))+d_\alpha(v,r(v))})/\eta$. Applying the triangle inequality and appropriate contraction bounds on $\alpha$ prior to this point gives us the result.

We are now ready to bound the expected expansion for each pair of nodes. This argument follows a similar line of reasoning to that for bi-Lipschitz extensions for $\ell_1$ in \cite{chawla2024}, but because we are using \onion partitions rather than CKR partitions, we are able to better bound the distortion in the ``bad case'' of the embedding.  

More specifically, the bad  case for the embedding is when a pair of nodes $x,y\in K$ is separated by the partition and they are much closer to each other than they are to $S$ (so routing their path through $S$ potentially causes a large distortion). In this situation, we bound the distance between $x$ and $y$ in terms of $d(x,S)$ and $d(y,S)$, and we use the fact that the probability of separation is bounded by $\rho\cdot \frac{d(x,y)}{d(\set{x,y},S)}$ to obtain our final bound. Crucially, the bound we obtain for the distance between $x\in K_{i}$ and $y\in K_{j}$ depends on walking from $x$ to $r(x)$ (the anchor point of $x$) through the embedding on $K_{i}^+$, then from $r(x)$ to $r(y)$ (the anchor point of $y$) through the embedding on $S$, and then from $r(y)$ to $y$ through the embedding on $K_{j}^+$ (see Figure \ref{fig:partition}). 
In the simplest sense, we could apply a bound of $c_S$ for the expansion of the distance between $r(x)$ and $r(y)$ and a bound of $c_k$ (the distortion bound for embedding subsets of size $k$) for the distances between $x$ and $s$ and between $y$ and $t$. Because we must pay $\rho$ as part of the separation probability in the partition, our overall expected distortion goes up to $O(\rho \cdot (c_S+c_k))$.\footnote{In fact, this bound applies even if we do not have a bounded diameter embedding (i.e. we just have an embedding with expected expansion $c_k$ on sets of size $k$) and even if the partitions we use do not meet the third criterion of \onion partitions (requiring  that the distance between a cluster node and $S$ be similar to the diameter of the cluster). }

However, because of the \onion partitions we use here, $d(x,s)$ is forced to be within a factor of $\gamma$ of the diameter $\Delta_x$ of the cluster $x$ is in. Further, because we use diameter-bounded embeddings on the set $K_{i}^+$, we know that the maximum distance between any pair of nodes in that embedding is at most $c_\Delta\cdot \Delta_{x}$. Thus, the distance between $x$ and $r(x)$ in this embedding is within a constant factor of their true distance, allowing us to reduce our expansion factor in this bad case to $O(\rho \cdot \gamma \cdot (c_S+c_\Delta))$. We potentially still incur a cost of $c_k$ if $x$ and $y$ are in the same component, so our overall cost is $O((c_S+c_\Delta)\cdot \rho\cdot \gamma+c_k)$, which gives us a potential savings if $c_S$ and $c_\Delta$ are much smaller than $c_k$ and $\gamma$ is a constant. Our analysis follows a similar reasoning to that of Calinescu, Karloff, and Rabani \cite{calinescu2005extension} for the $0$-extension problem, though we deal with an additional merge step and resulting effect on the embedding. 
\ifdagstuhl
See the full version of the paper for a full proof.
\else 
See Appendix \ref{sec:exp_contr_proofs} for a full proof.
\fi

\begin{figure}
    \centering
    \includegraphics[width=0.5\textwidth]{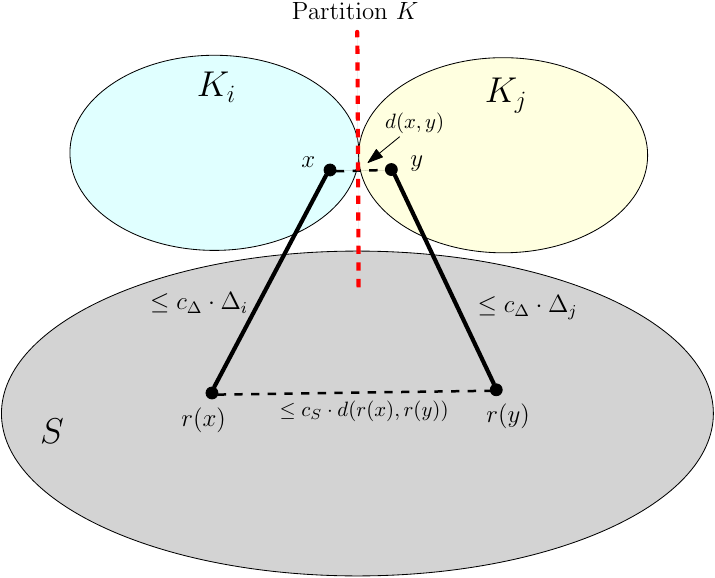}
    \caption{The "bad case" routing scenario for expansion bounds. Nodes $x \in K_i$ and $y \in K_j$ are physically close but separated by the onion terminal partition $\mathcal{K}$. The expected distance between $x$ and $y$ in the merged embedding is bounded by routing the path through their respective anchor points $r(x)$ and $r(y)$ within the base embedding $S$.}
    \label{fig:partition}
\end{figure}


\begin{lemma}[Expansion bound]\label{lem:expansion-general}
Let $X,S,\cD_S$ be as in \Cref{thm:merge-components}, and let $\cD$ be the distribution over embeddings from which \Cref{alg:extend} draws when given inputs $X,S,\cD_S$. Then for all $x,y\in X$ 
\begin{align*}
    \expec_{\alpha\sim \cD}[d_\alpha(x,y)] &\leq O(c_k+\gamma \cdot \rho \cdot (c_S+c_\Delta))\cdot d(x,y),
\end{align*}

where $c_k$ is the bound on the expansion of the bounded-diameter embeddings for each part of the partition used by the algorithm, $c_S$ is the expansion of $\cD_S$, $c_\Delta$ is the diameter expansion of the bounded-diameter embeddings used by the algorithm, and $k=|X\setminus S|$. 
\end{lemma}



Noting that \Cref{alg:extend} always outputs an embedding in which the distances on pairs of points in $S$ are the same as under the sampled embedding $\alpha_S$,
\Cref{thm:bilip-hst} (finding \bl extensions for HSTs) is immediate from \Cref{thm:merge-components,thm:frt,thm:onion,lem:merge-hst}.  
\Cref{thm:bilipl1} (finding \bl extensions for $\ell_1$) takes a little more work, and we prove it in \Cref{sec:merge}.

%% file: alg_extend.tex
\begin{algorithm}
\caption{{\bf Bi-Lipschitz extension}: Algorithm for finding a probabilistic bi-Lipschitz extension when the requirements of \Cref{thm:merge-components} are met.}
\label{alg:extend} 
\textbf{Input:} Metric space $(X,d)$, subset $S\subseteq X$, embedding $\cD_S:S\rightarrow\cY $; the spaces $X$ and $\cY$ should meet the conditions of \Cref{thm:merge-components}, and \pmerge\ will be the designated perfect merge function



\textbf{Output:} A randomly sampled embedding $\alpha:X\rightarrow \cY$.
\begin{algorithmic}[1]
\State Sample a partition $(\cK=\set{K_{i}}_{i\in\mathbb{N}},r)$ from an \change{\onion partition} of $K=X\setminus S$
\State Let $K_i^+$ be $K_i\cup \set{r(K_i)}$ for all $K_i\in \cK$
\State Sample $\alpha\sim \cD_S$
\For{each non-empty $K_{i}$}
\State Let $\cD_{i}$ be an expanding \change{bounded-diameter embedding} of $K_{i}^+$ into $\cY$
\State Sample $\alpha_{i}\sim \cD_{i}$
\State $\alpha\leftarrow \change{\pmerge(\alpha,\alpha_{i})}$
\EndFor
\State Return $\alpha$
\end{algorithmic}
\end{algorithm}

%% file: ckr_background.tex
\section{\Onion\ partitions}
\label{sec:partitions}

In this section we prove \Cref{thm:onion}. 

\Onionthm*

We begin with a discussion of the partition algorithm given by Calinescu, Karloff, and Rabani in their work on the $0$-extension problem \cite{calinescu2005extension}, which we refer to as CKR partitions.
Our \onion partitions will consist of CKR partitions intersected with a randomized bucketing scheme for $K=X\setminus S$. The CKR partitions ensure that nodes are separated with a probability proportional to their distance, and the randomized bucketing allows us to control the diameter of each cluster relative to the designated anchor point without increasing the probability of separation between $x$ and $y$ by too much.

Given a metric $(X,d)$ and a subset $S\subseteq X$, \cite{calinescu2005extension} finds a partition $\{K_1,K_2,\ldots,K_{|S|}\}$ of $K:=X\setminus S$ such that the ``anchor point'' of $K_s$ is $s\in S$, the probability of $x,y\in K$ being separated is $O(\log k\cdot \frac{d(x,y)}{d(\set{x,y},S)})$, and $x\in K_s$ for $s\in S$ implies $d(x,s)\leq 2d(x,S)$. 
Algorithm \ref{alg:ckr} presents a variation of the CKR algorithm in which we use a subset  
$T\subseteq S$ as the set of anchor points, rather than the entire set $S$.


\input{alg_ckr}

The CKR partitioning algorithm works by selecting a random permutation on some set of terminals and then ``assigning'' points in $X\setminus S$ to terminals in $T\subseteq S$ such that (1) a node in $X\setminus S$ is always assigned to a terminal that is fairly close to it, and (2) the probability that a pair of nodes $x,y\in X$ are placed in separate part of the partition is proportional to $d(x,y)/d(\set{x,y},T)$. We state these two properties as Observation \ref{obs:ckr-properties} and Lemma \ref{lem:ckr_separation}. Note that if we let the anchor point of a cluster $K_t$ for $t\in T$ be $t$ itself, then  these two facts imply that CKR partitions meet criteria (1) and (2) of \onion partitions, so we're left with fulfilling criterion (3).

\begin{observation}[Deterministic Proximity of CKR Rounding, implicit in \cite{calinescu2005extension}]\label{obs:ckr-properties}
For any $x\in X\setminus S$, if Algorithm \ref{alg:ckr} assigns $x$ to $K_t$ for some $t\in T$, then 
\begin{align*}
    d(x,t) &\leq \mu \cdot d(x,T) \leq 2\cdot d(x,T),
\end{align*}
where the last inequality holds because $\mu \in [1,2)$.
\end{observation}

\begin{lemma}[CKR Separation, Lemma 2.2 in \cite{calinescu2005extension}]\label{lem:ckr_separation}
Fix $x,y \in X$ and let $\delta := d(x,y)$. If $0 < \delta \le \frac{1}{4} \min\{d(x,T), d(y,T)\}$, then:
		\[
		\Pr[x,y \text{ are separated}] \le 4\mathcal{H}_{|T|} \left( \frac{\delta}{d(x,T)} + \frac{\delta}{d(y,T)} \right),
		\]
where $\mathcal{H}_m = 1 + \frac{1}{2} + \dots + \frac{1}{m}$ denotes the $m$-th harmonic number.
\end{lemma}

For our partitions, we will pick $T$ to be the set of ``nearest neighbors'' of nodes in $K$. That is, $T:=\set{t\in \arg\min_{s\in S}\set{d(s,x)}|x\in K}$. We assume that for each $x\in K$, we add only one of the nearest neighbors of $x$ to $T$ so that $|T|\leq k$. Moreover, by the choice of $T$, we also have $d(x,T)=d(x,S)$ for all $x\in K$.


\subsection*{Randomized bucketing}

Now we show how to refine the CKR partition technique so that the diameter of each partition is bounded by a constant factor times the distance from that partition to the terminal set $T$. In particular, we will use a technique called {\em randomized scale partitioning}. For any point $x \in K$, let $A_x = d(x,T)$ denote its distance to the nearest terminal.   We choose a shift parameter $u \in [0,1)$ uniformly at random and define $\beta=2^u$. We then 
(randomly) bucket $K$ into disjoint scales $B_i$ by distance from $T$:
\[
B_i := \{ x \in K \mid \beta \cdot 2^i \le A_x < \beta \cdot 2^{i+1} \}.
\]

Given a metric $(X,d)$ with diameter $\Delta$, subset $S\subseteq X$, and terminal set $T\subseteq S$, we will randomly partition $K=X\setminus S$ into sets $\set{K_t}_{t\in T}$ and buckets $\set{B_i}_{i\in [\lceil \log \Delta\rceil]}$,\footnote{Here, $[n]$ denotes $\set{0,1,2,\ldots,n}$.} where the $K_t$ are chosen by finding a CKR partition on $K$ with terminal set $T$ and the $B_i$ are chosen by randomized scale partition. Then we define a new partition
$\set{K_{t,i}}_{t\in T,i\in [\lceil \log \Delta\rceil]}$, where $K_{t,i}:=(K_t\cap B_i)$ and let the anchor point of this set be $t$ --- that is, $r(K_{t,i})=t$.

The primary properties of this partitioning that we will use are (1) for any node $x\in K_{t,i}$, the distance from $x$ to another node in $K_{t,i}$ is not too much larger than the distance from $x$ to $T$, and (2) the probability that a pair of nodes $x$ and $y$ are separated in this partition is not too much larger than the probability that $x$ and $y$ are separated under the CKR embedding. We express these as Observation \ref{lem:partition-diam} and Lemma \ref{lem:scale-separation}, respectively.

\begin{observation}[Partition diameter bound]\label{lem:partition-diam}
For any part $K_{t,i}$ of our scale partition and any node $x\in K_{t,i}$, if $\Delta_{t,i}$ is the diameter of $K_{t,i}$, then $d(x,t)\geq \Delta_{t,i}/6$


\end{observation}

\begin{proof}
For any pair $x,y\in K_{t,i}$, we have $d(x,y)\leq d(x,t)+d(t,y)\leq 2d(x,T)+2d(y,T)\leq 6d(x,T)$. The first inequality is by the deterministic properties of CKR partitions, and the second is because $x$ and $y$ being in the same bucket implies $d(y,T)$ is no more than twice $d(x,T)$. Thus, $\Delta_{t,i}\leq 6\cdot d(x,T)\leq 6\cdot d(x,t)$.
\end{proof}

\begin{lemma}[Scale Partition Separation]\label{lem:scale-separation}

For any pair of points $x,y\in K=X\setminus S$, let $\delta:=d(x,y)$, and assume without loss of generality that $A_x \le A_y$. The probability that $x$ and $y$ assigned to distinct buckets $B_i$ and $B_j$ is bounded by $ O\left(\frac{\delta}{A_x}\right)$.

Thus, by the union bound and Lemma \ref{lem:ckr_separation}, the probability that $x$ and $y$ are in different parts of the partition $\set{K_{t, i}}$ is at most $O(\log |T|\cdot \frac{\delta}{A_x})$.

\end{lemma}
	
\begin{proof}
Let $x$ be in bucket $i$ and $y$ be in bucket $j$.
The scale boundaries are determined by $\beta 2^i = 2^{i+u}$, which is equivalent to placing a grid of integers shifted by a uniformly random $u \in [0,1)$ on a logarithmic scale. The probability that an integer boundary falls strictly between $\log_2 A_x$ and $\log_2 A_y$ is bounded by the length of the interval between their logarithmic distances. By the metric triangle inequality, $A_y \le A_x + \delta$. Therefore:
\begin{align*}
    \Pr[i \neq j] &\le \log_2(A_y) - \log_2(A_x) \\
    &= \log_2\left(\frac{A_y}{A_x}\right) \\
    &\le \log_2\left(\frac{A_x + \delta}{A_x}\right) = \log_2\left(1 + \frac{\delta}{A_x}\right) \\
    &\le \frac{\delta}{A_x \ln 2} \;=\; O\left(\frac{\delta}{A_x}\right)
\end{align*}
where the final inequality follows from the standard logarithmic bound $\ln(1+z) \le z$ for all $z \ge 0$.
\end{proof}


Thus, the partitions we have constructed inherit property (2) of \onion partitions with closeness $2$ from the CKR bounds on anchor point distances and our choice of $T$ (as the nearest neighbors of nodes in $K$), and they inherit property (3) of \onion partitions with \band $6$ from bounds on the diameter of randomized bucketing (and the fact that taking a subset of a cluster never increases its diameter). They also achieve property (1) with parameter $O(\log k)$ by \Cref{lem:ckr_separation,lem:scale-separation}, our choice of $T$, and a union bound. Thus, we get \Cref{thm:onion}.


%% file: alg_ckr.tex
\begin{algorithm}
\caption{{\bf CKR partitioning}}
\label{alg:ckr}
\textbf{Input:} Metric space $(X,d)$, subset $S\subseteq X$, a set of terminals $T\subseteq S$

\textbf{Output:} A partition $\set{K_t}_{t\in T}$ of $X\setminus S$ 
\begin{algorithmic}[1]
\State Sample a permutation $\pi$ uniformly at random over $T$
\State Sample $\mu\in [1,2)$ uniformly at random. 
\For{Each node $x\in X\setminus S$}
\State Let $t\in T$ be the first terminal under the ordering $\pi$ such that $d(x,t)\leq \mu\cdot d(x,T)$
\State Add $x$ to $K_t$
\EndFor
\State Return $\set{K_t}_{t\in T}$, where the anchor point of $K_t$ is set to $t$
\end{algorithmic}
\end{algorithm}

%% file: merge_functions.tex
\section{Merge functions for HSTs and \texorpdfstring{$\ell_1$}{l1}}\label{sec:merge}

In this section, we give explicit perfect merge functions for $\ell_1$ and HSTs, and we prove \Cref{thm:bilipl1}.

\subsection{Merge for \texorpdfstring{$\ell_1$}{l1}}

\input{alg_l1merge}

We begin by presenting a perfect merge function for $\ell_1$ embeddings. This is captured in \Cref{alg:mergel1}, which essentially just extends each input embedding to the entire set $Z_1\cup Z_2$ by mapping everything not already in the domain to the same point as $u\in Z_1\cap Z_2$ and concatenating the resulting embeddings together. This is easily seen to meet the requirements of a perfect merge function with contraction bound $1$ and thus we get \Cref{lem:merge-l1} which states that $\ell_1$ has perfect merge functions.
We now move on to proving \Cref{thm:bilipl1}.

\biliplonethm*

\begin{proof}
Part (1) of the proof is almost immediate from \Cref{thm:merge-components,thm:frt,thm:onion}, the fact that $\ell_1$ has a perfect merge function with contraction factor $1$, and the fact that \Cref{alg:extend} outputs an embedding in which pairs of nodes in $S$ have the same distance as under $\alpha_S$. The problem is that \Cref{thm:merge-components} implies the existence of a {\em probabilistic} extension of the original embedding. However, for any finite support probabilistic embedding $\cD$ of $X$ into $\ell_1$, there exists a deterministic embedding $\alpha:X\rightarrow \ell_1$ with $d_\alpha(x,y)=\expec_{\alpha'\sim \cD}[d_{\alpha'}(x,y)]$. (Note that the dimension of $\alpha$ may be much larger than those of the embeddings in the support of $\cD$, and it is not clear how to efficiently find $\alpha$ when $\cD$ has a large support.) See e.g. \cite{chawla2024} for more.
Proving the second part of this theorem requires a bit more of a sophisticated de-randomization method, and we defer it to
\ifdagstuhl
the full version of the paper.
\else 
Appendix \ref{sec:bilipl1-lastproof}.
\fi
\end{proof}

\subsection{Merge for HSTs}

We now present a \goodmerge\ merge function for HSTs to prove \Cref{lem:merge-hst}, which states that HSTs have perfect merge functions with contraction factor $2$.

Algorithm \ref{alg:merge} takes in two $\beta$-HSTs for a fixed ratio $\beta$ and returns another $\beta$-HST. 
This algorithm works by ``combining'' the ancestors of  the common point $u$ at each level of the two HSTs. That is, the new HST consists of $u$ and all its ancestors, and each ancestor of $u$ has as its child subtrees the child subtrees of the ancestor of $u$ at its same level in each HST. \Cref{fig:merge} gives a visual depiction of this algorithm.
We defer the correctness proof for \Cref{lem:merge-hst} to
\ifdagstuhl
the full version of the paper
\else 
Appendix \ref{sec:hstmergeproof} 
\fi 
and briefly outline the idea here. 

For any $x,y\in Z_1$ or $x,y\in Z_2$, the least common ancestor of $x$ and $y$ must be at the same level in the new tree as it was in the original embedding of $Z_1$ or $Z_2$, so their  distance is unchanged from the original embedding.
To bound the distance between $x\in Z_1,y\in Z_2$, consider the least common ancestor of $x$ and $y$ in the HST produced at the algorithm. The first ancestor of $y$ that has descendants from $Z_1$ is the least common ancestor of $y$ and $u$. Thus, the least common ancestor of $x$ and $y$ cannot be lower than the least common ancestor of $y$ and $u$, so the distance between $x$ and $y$ is at least that between $y$ and $u$ under the original embedding on $Z_2$. Likewise, the distance between $x$ and $y$ is at least the distance between $x$ and $u$ under the original embedding on $Z_1$, so we get contraction at most $2$.



\begin{figure}
    \centering
    \includegraphics[width=0.8\linewidth]{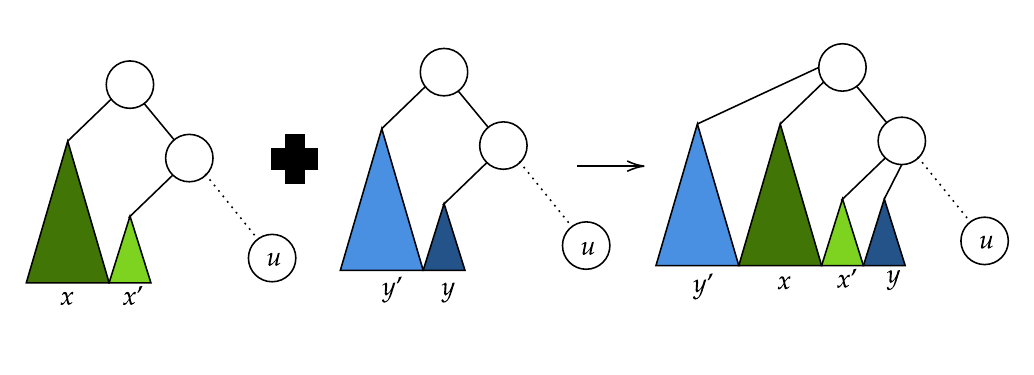}
    \caption{This figure depicts what happens in the algorithm $\merge$. In particular, $u$ is the node common to both HSTs on the left side of the figure. The triangles represent subtrees, and the letters below the triangles indicate the presence of the corresponding node as a leaf of that subtree. Note that after merging trees, the least common ancestor of $x$ and $y$ is that of $x$ and $u$ from the first tree, as this was higher than that of $y$ and $u$.
    This is the foundational idea behind why  $\dista[\alpha](x,y)\geq \max\set{\dista[\alpha_1](u,x),\dista[\alpha_2](u,y)}$ holds when $\alpha$ is the output embedding for $\merge$ on inputs $\alpha_1$ and $\alpha_2$.  }
    \label{fig:merge}
\end{figure}

\input{alg_merge_hst}

%% file: alg_l1merge.tex
\begin{algorithm}
\caption{$\mergel(\alpha_1,\alpha_2)$ Algorithm for merging two $\ell_1$ embeddings}
\label{alg:mergel1}
\textbf{Input:} Embeddings $\alpha_1:Z_1\rightarrow \ell_1$, $\alpha_2:Z_2\rightarrow \ell_1$, where $|Z_1\cap Z_2|=1$. \\
\textbf{Output:} { An embedding of $Z_1\cup Z_2$ into $\ell_1$} 
\begin{algorithmic}[1]
\State Let $u=Z_1\cap Z_2$
\State For all $x\in Z_1$, let $\alpha(x)\leftarrow \alpha_1(x)\circ \alpha_2(u)$, where $\circ$ is concatenation
\State For all $x\in Z_2$, let $\alpha(x)\leftarrow \alpha_1(u)\circ\alpha_2(x)$
\State Return $\alpha$
\end{algorithmic}
\end{algorithm}

%% file: alg_merge_hst.tex
\begin{algorithm}
\caption{$\merge(\alpha_1,\alpha_2)$ Algorithm for merging two $\beta$-HSTs}
\label{alg:merge}
\textbf{Input:} Embeddings $\alpha_1:Z_1\rightarrow T_1$, $\alpha_2:Z_2\rightarrow T_2$ for $\beta$-HSTs $T_1,T_2$, where $|Z_1\cap Z_2|=1$. {We assume without loss of generality that the height of $T_1$ is {at least as large as} the height of $T_2$.}\\
\textbf{Output:} { An embedding of $Z_1\cup Z_2$ into a single $\beta$-HST} 
\begin{algorithmic}[1]
\State Let $u=Z_1\cap Z_2$
\State Create a new HST $T$ that is initially a copy of $T_1$
\State $\alpha(x)\from \alpha_1(x)$ for all $x\in Z_1$
\For{$i=0$ to the height of {$T_1$}} 
\For{the subtree of each child of $\alpha_2(u)$'s ancestor  at height $i$ in $T_2$} 
\If{ the subtree does not contain $\alpha_2(u)$} \Comment{{\scriptsize \change{i.e. if the subtree is not yet merged into $T_1$}}}
\State Copy the subtree and make it a child of {$\alpha_1(u)$}'s ancestor at level $i$ in $T$
\For{any $x\in Z_2\setminus \set{u}$ such that $\alpha_2(x)$ is in the subtree at this step }
\State Set $\alpha(x)$ to be the copy of $\alpha_2(x)$ created in $T$ at this step
\EndFor
\EndIf
\EndFor
\EndFor
\State Return $\alpha$
\end{algorithmic}
\end{algorithm}

%% file: hst.tex
\section{Outlier embeddings for HSTs\label{sec:hst}}

Now, given the existence of \bl extensions for HSTs, we turn to approximating outlier embeddings into HSTs.  In \cite{chawla2024}, this type of approximation is obtained for $\ell_2$ by introducing new ``outlier variables'' into a semidefinite program for a minimum distortion $\ell_2$ embedding and rounding. We will employ a similar technique here using the linear program  and associated rounding algorithm of \cite{mungala2023} to approximate the optimal distortion for probabilistic HST embeddings.

Let $\Delta$ be the diameter of $(X,d)$ and let $M=\set{2^0,2,2^2,\ldots,2^{\lceil \log\Delta \rceil}}$. Let $\zeta$ be the constant factor from Theorem \ref{thm:bilip-hst} and let $B_i(r):=\set{j\in X: d(i,j)\leq r}$. We define (\ref{eqn:hst1}), which is parameterized by inputs $((X,d),k,c)$. The changes from \cite{mungala2023} are highlighted in red.

Table \ref{tab:variables} gives an intuitive description of the variables in the LP.
Note that in our LP $c$ is a parameter, but in \cite{mungala2023} it is the variable they are minimizing.  The biggest change from \cite{mungala2023} is that we have introduced new variables $\delta_i$ for each $i\in X$ that we use as an indicator for whether $i$ is an outlier in the embedding.

\begin{align}
     \min \change{\sum_{i=1}^n \delta_i}
    \label{eqn:hst1}\tag{HST LP}  
    \\
    \forall j,j'   \sum_{r\in M} r\cdot \gamma_{jj'}^r &\leq \change{(4+\zeta\cdot (\delta_j+\delta_{j'})\cdot \log k)}\cdot c\cdot d(j,j') \label{eqn:hst2} \\
    \forall i,r,(j,j'\in B_i(r)) \ \ \ \ \  \min(x_{i,j}^r,x_{ij'}^r) &\geq z_{ijj'}^r  \label{eqn:hst3}\\
     \forall j,j',i,r \sum_{i:j,j'\in B_i(r)}z_{ijj'}^r &\geq 1-\gamma_{jj'}^r \label{eqn:hst4}\\
    \forall j,r \sum_{i: j\in B_i(r)}x_{ij}^r &=1 \label{eqn:hst5} \\
    \forall j,j',(r< d(j,j'))  \ \ \ \ \  \gamma_{jj'}^r &= 1\label{eqn:hst7} \\ 
    \forall j,j',r  \ \ \ \  \gamma_{jj'}^r &\leq 1 \label{eqn:hst8} \\
    \forall i,j,j',r \ \ \ \ \ \gamma_{jj'}^rx_{ij}^r,z_{ijj'}^r,\change{\delta_i} &\geq 0 ; \change{\delta_i\leq 1}\label{eqn:hst9}
\end{align}


\input{var_table}

Given an optimal solution to this linear program, we need to define an outlier embedding with low distortion and few outliers. To do so, we will first use the rounding algorithm of \cite{mungala2023} to find an embedding on the entire metric, and then we take as outliers all vertices $i$ such that $\delta_i\geq \delta\st$ for some threshold $\delta\st$.  In 
\ifdagstuhl
the full version of the paper
\else 
Appendix \ref{sec:proofoutlier} 
\fi
we formally prove that this rounding algorithm works and gives us Theorem \ref{thm:outlier} as a result. In a similar manner to $\ell_2$ outlier embeddings in \cite{chawla2024}, we use the existence of nested embeddings to bound the optimal LP solution and thus obtain our final result.

Further, if we modify (\ref{eqn:hst1}) by changing the objective function to $\sum_iw_i\cdot \delta_i$ for some set of weights $w_i$, we can use the same technique to obtain a bicriteria approximation for the weighted variant of the outlier embedding problem, where the cost of making $i$ an outlier is $w_i$. This result is expressed as Corollary \ref{thm:weighted}, and we prove it in 
\ifdagstuhl
the full version of the paper.
\else 
Appendix \ref{sec:weighted}.
\fi

%% file: var_table.tex
\begin{table}
    \centering
    \begin{tabular}{c|c}
        variable & intuitive meaning \\
        \hline 
         $\delta_i$ & indicator for whether $i$ is an outlier \\
         $z_{ijj'}^r$ & the probability that $i$ is the representative of $j$ and $j'$ at level $r$ \\
         $\gamma_{jj'}^r$ & the probability that $j$ and $j'$ are separated at level $r$ \\
         $x_{ij}^r$ & the probability that $i$ is the representative of $j$ at level $r$\\
         \hline 
    \end{tabular}
    \caption{We imagine assigning a ``representative'' at each node of the HST such that the representative of a leaf node is the leaf node itself, and the representative of any other node is the same as one of its children's representatives. This table describes the intuitive meaning of each variable in the LP (which is largely the same as in \cite{mungala2023}). }
    \label{tab:variables}
\end{table}

%% file: applications.tex
\section{Applications\label{sec:applications}}



In the Minimum Cost Communication Tree (MCCT) problem, we are given a finite metric $\cX=(X,d)$ and weights $r_{ij}$ for all pairs of nodes $i,j$ and we are asked for a non-contracting embedding $\alpha$ of $X$ into a tree with vertex set $X$ such that $\sum_{i,j\in X}r_{ij}\cdot d_\alpha(i,j)$ is minimized. 
This problem is motivated by evolutionary biology, in which a scientist may attempt to reconstruct a tree of relationships between multiple organisms, and she comes to the problem with some known minimum distances between them, derived through a method such as DNA sequence analysis. In such a setting, the $r_{ij}$ are an additional tunable parameter that can allow the scientist to indicate pairs she considers ``more likely'' to be related, by giving them a bigger $r_{ij}$ value. 
As noted by Wu et al. \cite{wu2000}, this problem can be approximated in expectation to within a factor of $O(\log n)$ by simply embedding randomly into an HST and converting it to a tree with vertex set $X$.
Directly applying the algorithm of Munagala et al. \cite{mungala2023} with the observations of Wu et al. \cite{wu2000}, we have a randomized algorithm that gets an $O(c(\cX))$-approximation for the MCCT problem, where $c(\cX)$ is the minimum distortion needed to probabilistically embed $\cX$ into  HSTs.

It is quite natural to consider an outlier version of the MCCT problem. In the evolutionary biology setting, we could consider the case that we are given some noisy data points, or points that don't really belong in the set being analyzed. 
One approach to this type of outlier problem may be to look for a $k$ outlier tree embedding of minimum distortion, but the set of outliers that minimizes the embedding distortion does not necessarily also minimize the MCCT solution cost.
We may expect that the outlier MCCT problem is easier to solve on tree metrics than on general metrics, so we could instead try the following approach: embed the input metric into a tree, and {\em then} try to find the best outliers for the resulting tree metric.  In fact, if no outliers are allowed, MCCT is trivial on tree metrics. 
Further, in Subsection \ref{sec:mcctontree} we give a bicriteria approximation to the $k$-outlier MCCT problem when the input is a tree metric. In particular, if the optimal cost of a $k$ outlier solution for a given instance of a tree metric is $OPT_k$, then we give an algorithm that finds a solution with at most $3k$ outliers and at most $24OPT_k$ cost. 
Let $c_k(\cX)$ be the minimum distortion required for probabilistically embedding all but at most $k$ points in $\cX$ into HSTs.
We can use the outlier embedding algorithm from this paper to embed any input metric into a tree with outliers and then apply the approximation algorithm for outlier MCCT on trees, obtaining a solution with at most $O(\frac{k}{\epsilon}\log k)$ outliers and expected cost $O((1+\epsilon)\cdot c_k(\cX)\cdot OPT_k)$. By contrast, directly using the embedding algorithm of \cite{mungala2023} with the tree approximation algorithm from this section instead gives a solution with at most $3k$ outliers and $O((1+\epsilon)\cdot c_0(\cX)\cdot OPT_k)$ expected cost, so we are unable to  save on the cost of the optimal solution if there is a large subset that embeds well into HSTs but the entire metric does not. We explain this further in 
\ifdagstuhl
the full version of the paper. 
\else 
Appendix \ref{sec:extra-app}.
\fi

\subsection{MCCT with outliers\label{sec:mcctontree}}

Let $(T,d)$ be a tree metric, and let $k$ be an integer. Then we define variables $x_{ij}$ for each pair of nodes $ij$ and variables $\delta_i$ for each variable $i$. We can use them to make the following linear program.

\begin{align*}
    &\min \sum_{i,j\in T} x_{ij} \cdot r_{ij} \cdot d(i,j) \quad s.t \label{eqn:mcctout}. \tag{MCCT Outlier LP} \\
    \forall i,j\in T: \quad &\delta_i+\delta_j+x_{ij}\geq 1 \\
     & \ \ \ \ \ \ \ \ \ \sum_{i\in T}\delta_i\leq k \\
    \forall i,j\in T: \quad & \delta_i,x_{ij}\in [0,1]
\end{align*}

\begin{claim}\label{claim:mcct-out}
Let $(T,d)$ be a tree metric such that there exists 
a subset $K\subseteq T$, $|K|\leq k$, and $\sum_{i,j\in T\setminus K}r_{ij}\cdot d(i,j)\leq D$.

Then the following procedure finds a subset $K'\subseteq T$ with $|K'|\leq 3k$ and an embedding $\alpha$ of $T\setminus K'$ into a tree with the same vertex set  and $\sum_{i,j\in T\setminus K'}d_{\alpha}(i,j)\leq 24D$.

\begin{enumerate}
    \item Solve \ref{eqn:mcctout} and let $x_{ij},\delta_j$ be the resulting values
    \item Let $K'=\set{i:\delta_i\geq 1/3}$
    \item Embed $T\setminus K'$ into a tree $T'$ with distortion at most $8$, using the algorithm of Gupta \cite{gupta2001steiner}
\end{enumerate}
\end{claim}

\begin{proof}
First, \ref{eqn:mcctout} has a solution of value at most $D$. In particular, let $K$ be as in the claim statement and consider setting $\delta_i=\begin{cases}
    1 & i\in K \\
    0 & else
\end{cases}$ and $x_{ij}=\begin{cases}
    1 & \delta_i=\delta_j=0 \\
    0 & else
\end{cases}$. As there are only $k$ outliers, the second constraint is clearly satisfied, and we have defined $x_{ij}$ such that the other constraints are also satisfied. Further, $x_{ij}$ is $0$ if either $i$ or $j$ is an outlier in $K$, so the objective function value is just $\sum_{i,j\in T\setminus K}r_{ij}d(i,j)$, which by the claim statement is bounded by $D$. 

Let $K'$ be as in the claim. As the $\delta_i$ sum to at most $k$, we know $|K'|\leq 3k$. Consider any $x_{ij}$ with $i,j\notin K'$ We have $x_{ij}\geq 1-\delta_i-\delta_j\geq 1/3$ by how we chose $K'$. The cost associated with a given pair $ij$ of non-outliers in our solution is $r_{ij}d(i,j)$, but that same pair contributed at least $r_{ij}d(i,j)/3$ to the optimal objective function value, so our final cost is at most $3$ times the objective, which is at most $3D$. At this point, we have an isometric (distortion $1$) embedding of the metric induced by $T\setminus K'$ into the original metric $T$ itself. Thus, we can view $K'$ as a set of Steiner points in our original metric.

Finally we consider the loss incurred applying the algorithm of Gupta \cite{gupta2001steiner}. This creates a tree embedding that removes Steiner points and incurs distortion at most $8$, so the cost of this embedding on $T\setminus K'$ is at most $24D$, as desired. This step is what prevents us from applying this algorithm directly to general metrics. 
\end{proof}



%% file: put_together_proofs.tex
\section{Missing proofs from \texorpdfstring{\Cref{sec:components}}{Section \ref*{sec:components}}}
\label{sec:exp_contr_proofs}

In this appendix, we give the proofs of \Cref{lem:contraction-general,lem:expansion-general}.

\subsection{Proof of \texorpdfstring{\Cref{lem:contraction-general}}{Lemma \ref*{lem:contraction-general}}}

\begin{proofof}{\Cref{lem:contraction-general}}
We first consider when the perfect merge function has a contraction bound $\eta$. Recall from \Cref{fact:induct} that for any pair of nodes $x,y$ we only need to consider the first time when $x$ and $y$ are both in the domain of $\alpha$. We consider three cases.

\begin{enumerate}
    \item If $x,y\in S$, they both enter the domain of $\alpha$ when it is set to be $\alpha_s\sim \cD_S$. As $\cD_S$ is expanding, $d_{\alpha_S}(x,y)\geq d(x,y)$, so the nodes incur no contraction.

    \item If $x\in S,y\in K$, $y$ must be in some cluster $K_i$ of the sampled \onion partition and have some assigned anchor point $r(y)$. Fix $i$ and $r(y)$. For this choice of partition, $x$ and $y$ first appear in the domain of $\alpha$ together when we merge $\alpha$ (the embedding up to this point in the algorithm) and $\alpha_i$ (the embedding we sample for $K_i^+$). Let $\alpha'$ be the embedding after this merge step and let $\alpha$ be the embedding immediately before. We have 
    \begin{align*}
        d_{\alpha'}(x,y) &\geq \frac{d_\alpha(x,r(y))+d_{\alpha_i}(r(y),y)}{\eta } \\
        &\geq \frac{d(x,r(y)) + d(r(y),y)}{\eta} \\
        &\geq \frac{d(x,y)}{\eta},
    \end{align*}

    where the first line is by the properties of the merge function; the second is by the fact that $\alpha_i$ is expanding and because $x,r(y)\in S$ so case 1 applies; and the third is by the triangle inequality on $d$.

    \item Finally, consider $x,y\in K$. Fix a sample $(\cK,r)$ from an \onion partition and let $r(x),r(y)$ be the assigned anchor points for $x,y$ in that partition. We divide into two subcases:
    \begin{itemize}
        \item If $x,y\in K_i$ for some $K_i\in \cK$, they both enter the domain of $\alpha$ when we sample an expanding embedding $\alpha_i$ on $K_i^+$ and merge it into $\alpha$. By the fact that $\alpha_i$ is expanding and the properties of perfect merge functions, $d_\alpha(x,y)\geq d(x,y)$. 

        \item Otherwise, $x\in K_i,y\in K_j$ for some $i\neq j$. Assume wlog that $K_j$ is processed by the algorithm after $K_i$. Then $y$ first enters the domain of $\alpha$ when $K_j$ is processed and the sampled embedding $\alpha_j$ on $K_j^+$ is merged with $\alpha$. Let $\alpha$ be the embedding before this merge step (which includes $x$ in its domain), and let $\alpha'$ be the embedding after this merge step. 

        \begin{align*}
            d_{\alpha'}(x,y) &\geq \frac{d_\alpha(x,r(y))+d_{\alpha_j}(r(y),y)}{\eta} \\
            &\geq \frac{d(x,r(y))/\eta + d(r(y),y)}{\eta } \\
            &\geq \frac{d(x,r(y))+d(r(y),y)}{\eta^2}
        \end{align*}

    \end{itemize}

\end{enumerate}
\end{proofof}

\subsection{Proof of \texorpdfstring{\Cref{lem:expansion-general}}{Lemma \ref{lem:expansion-general}} }

\begin{proofof}{\Cref{lem:expansion-general}} 
We first analyze the most general case where $x,y\in K$. Fix $x,y$ and denote $\delta = d(x,y)$.
Wlog we assume $d(x,S)\leq d(y,S)$. 

We split the analysis into two cases based on the distance $\delta$ relative to the distance from $x$ to $S$:
\begin{description}

\item[{Case 1: Large distance}]($\delta > \frac{1}{4} d(x,S)$).

By the triangle inequality, $d(y,S)\leq d(x,S)+\delta<5\delta$. Thus, $d(r(x),r(y))\leq d(r(x),x)+d(x,y)+d(y,r(y))\leq 2d(x,S)+2d(y,S)+\delta\leq 4d(x,S)+3\delta<19\delta$ (where the last inequality is by the fact that $\delta>\frac{1}{4}d(x,S)$ in this case).

Let $(\cK,r)$ be a random variable equal to the partition chosen by the \onion partition, and let $\Pr[(\cK,r)]$ be the probability that this partition is chosen. Let $r(x)$ be the anchor point of $x$'s cluster in $\cK$ and let $r(y)$ be the same for $y$. 
Fix a choice of $(\cK,r)$ and note that once we have done so $r(x)$ and $r(y)$ are also fixed to some points, say $s$ and $t$. Further, let $\cD_x$ be the bounded-diameter embedding we use on $x$'s cluster in $\cK$, and let $\cD_y$ be the same for $y$.
If $\alpha$ is the embedding returned by the algorithm, we get

\begin{align*}
    \expec_{\alpha\sim \cD}[d_\alpha(x,y)|(\cK,r)] 
    &\leq \expec_{\alpha\sim \cD}[d_\alpha(x,s)|(\cK,r)]+\expec_{\alpha\sim \cD}[d_\alpha(s,t)|(\cK,r)]+\expec_{\alpha\sim \cD}[d_\alpha(t,y)|(\cK,r)] \\ 
    &\leq \expec_{\alpha\sim \cD_x}[d_\alpha(x,s)|(\cK,r)] +\expec_{\alpha\sim \cD_S}[d_\alpha(s,t)|(\cK,r)]+\expec_{\alpha\sim \cD_y}[d_\alpha(t,y)|(\cK,r)]  \\
    &\leq c_k\cdot d(x,s)+c_S\cdot d(s,t)+c_k\cdot d(t,y) \\
    &\leq 2c_k\cdot (d(x,S)+d(y,S))+c_S\cdot d(s,t)  \\
    &\leq 2c_k\cdot (4\delta + 5\delta )+c_S\cdot 19\delta \\
    &\leq (18c_k+19c_S)\cdot \delta.
\end{align*}

The first line is by the triangle inequality on $\alpha$. The second line is by the properties of perfect merge functions and the fact that $x$ and $s$ first appear in the domain of the same embedding when we draw an embedding from $\cD_x$.\footnote{Technically, the embedding $\cD_x$ is on a set of size at most $k+1$, not at most $k$ (since it must embed the anchor point as well as the cluster). However, up to constant factors, $c_k$ is no smaller than $c_{k+1}$ and we write it this way for convenience. To see this, note that given an embedding on $k$ points, if the target metric has a perfect merge function we can add a single point to the embedding while blowing up the worst-case distortion and the diameter expansion by at most a constant factor. Consider simply finding an embedding of the new point and its closest neighbor among the original $k$ points and merging that embedding with the original embedding on $k$ points. .} (Recall that if we have fixed $(\cK,r)$, then $\cD_x$ is also fixed, so we can take an expectation over it.) The third line is because while $(\cK,r)$ determines the distributions $\cD_x$ and $\cD_y$, the randomness used in sampling from each of these embeddings is independent of $(\cK,r)$ (as is the randomness in sampling from $\cD_S$). The fourth line is by the second property of \onion partitions ($x$ is close to its anchor point). The rest are  by the bounds we established at the beginning of this case.

By the law of total expectation, we have $\expec_{\alpha \sim \cD}[d_\alpha(x,y)]=\sum_{(\cK,r)}\Pr[(\cK,r)]\cdot \expec_{\alpha\sim \cD}[d_\alpha(x,y)|(\cK,r)]$, so we get our bound immediately.

\item[{Case 2: Small distance}]($\delta \le \frac{1}{4} d(x,S)$).

We begin by bounding the distance between $x$ and $y$ with respect to a specific partition.
As before, fix a choice of partition $(\cK,r)$ and let $s$ be the anchor point of $x$'s cluster under this partition, let $t$ be the anchor point of $y$'s cluster under this partition, and let $\cD_x$ be the bounded-diameter embedding corresponding to $x$'s cluster in $\cK$, plus its anchor point, and let $\cD_y$ be the same for $y$. Further, let $\Delta_x$ be the diameter of $x$'s cluster in $\cK$, and let $\Delta_y$ be analogous for $y$.

We have 

\begin{align*}
\expec_{\alpha\sim \cD}[d_\alpha(x,y)|(\cK,r)]&\leq \expec_{\alpha\sim \cD}[d_\alpha(x,s)|(\cK,r)]+\expec_{\alpha\sim\cD}[d_\alpha(s,t)|(\cK,r)]+\expec_{\alpha\sim \cD}[d_\alpha(t,y)|(\cK,r)]\\
&\leq \expec_{\alpha\sim \cD_x}[d_\alpha(x,s)|(\cK,r)]+\expec_{\alpha\sim\cD_S}[d_\alpha(s,t)|(\cK,r)]+\expec_{\alpha\sim \cD_y}[d_\alpha(t,y)|(\cK,r)] \\
&\leq c_\Delta \cdot \Delta_x + c_S\cdot d(s,t) + c_\Delta \cdot \Delta_y \\
&\leq \gamma \cdot c_\Delta\cdot d(x,s) + c_S\cdot d(s,t)+\gamma \cdot c_\Delta\cdot d(y,t) \\
&\leq 2\gamma \cdot c_\Delta \cdot d(x,S) + 2\gamma \cdot c_\Delta\cdot d(y,S) + c_S\cdot d(s,t).
\end{align*}

The first line is by the triangle inequality, the second is by the properties of perfect merge functions, and the third is due to the independence between the sampled distribution $\cD_S$ and the partition $(\cK,r)$ and the expected diameter expansion bound of the diameter-bounded embeddings. The fourth line is by the properties of \onion partitions with \band $\gamma$, and the fifth is by the properties of \onion partitions with closeness $2$.

Note that since $\delta \le \frac{1}{4}d(x,S)$, we know $d(y,S) \le d(x,S) + \delta \le \frac{5}{4}d(x,S) = O(d(x,S))$. Similarly, $d(s,t) \le 2d(x,S) + \delta + 2d(y,S) \le \frac{19}{4}d(x,S) = O(d(x,S))$. Thus, we have shown that \fbox{for any $(\cK,r)$, $\expec_{\alpha\sim \cD}[d_\alpha(x,y)|(\cK,r)]\leq O(\gamma \cdot (c_S+c_\Delta)\cdot d(x,S))$}.\footnote{Alternatively, $\expec_{\alpha\sim \cD}[d_\alpha(x,y)|\cK]\leq O((c_S+c_k)\cdot d(x,S))$, with no dependence on $\gamma$ or $c_\Delta$. }

Further, note that if $x$ and $y$ are {\em not} separated by $(\cK,r)$ (i.e. they're in the same cluster in $\cK$), then $\expec_{\alpha\sim \cD}[d_\alpha(x,y)|(\cK,r)]$ is bounded by the expected distance between $x$ and $y$ in the embedding drawn from $\cD_x$ (which $x$ and $y$ must both be the in the domain of). Using bounds on our bounded-diameter embedding and the independence of this draw from $(\cK,r)$, we get that  \fbox{for any $(\cK,r)$ that does not separate $x,y$, $\expec_{\alpha\sim \cD}[d_\alpha(x,y)|(\cK,r)]\leq c_k\cdot d(x,y)$}.


Let $K_{sep}$ be the set of partitions that separate $x$ and $y$ and let $K_{\neg sep}$ be the set of partitions that do not separate $x$ and $y$. Further, let $\phi$ be the event that $x$ and $y$ are separated by the partition.
Combining everything so far, we get

\begin{align*}
    \expec_{\alpha\sim \cD}[d_\alpha(x,y)] &= \sum_{(\cK,r)}\Pr[(\cK,r)]\cdot \expec_{\alpha\sim \cD}[d_\alpha(x,y)|(\cK,r)] \\
    &= \sum_{(\cK,r)\in K_{\neg sep}} \Pr[(\cK,r)]\cdot \expec_{\alpha\sim \cD}[d_\alpha(x,y)|(\cK,r)] 
    \\
    & \ \ \ \ \ + \sum_{(\cK,r)\in K_{sep}}\Pr[(\cK,r)]\cdot \expec_{\alpha\sim \cD}[d_\alpha(x,y)|(\cK,r)] \\
    &\leq c_k\cdot d(x,y) \cdot \sum_{(\cK,r)\in K_{\neg sep}} \Pr[(\cK,r)] \\
    & \ \ \ \ \ + O((c_S+c_\Delta)\cdot d(x,y))\cdot \sum_{(\cK\,r)\in K_{sep}}\Pr[(\cK,r)] \\
    &= c_k \cdot d(x,y)\cdot \Pr[\neg \phi] + O((c_S+c_\Delta)\cdot d(x,S)) \cdot \Pr[\phi] \\
    &\leq 1\cdot c_k\cdot d(x,y)+ O(\gamma \cdot (c_S+c_\Delta)\cdot d(x,S))\cdot \left(\rho\cdot  \frac{d(x,y)}{d(x,S)}\right) \\
    &\leq O( (c_k+\gamma \cdot (c_S+c_\Delta)\cdot \rho)\cdot d(x,y)).
\end{align*}
 
The first line is by the law of total expectation, the third is by our earlier bounds on the expected distance between $x$ and $y$ when they are separated versus when they are not, and the fifth is by the first property of \onion partitions (the bound on the probability of separation).\footnote{Note that if we want to eliminate dependence on $\gamma$ and $c_\Delta$, we get a bound of $O((c_S+c_k)\cdot \rho \cdot d(x,y))$.} 

\end{description}

The cases where at least one point is in $S$ follow directly from this analysis. If $x\in S,y\in K$, this can be analyzed identically to adding a point in $K$ at distance $0$ from $x$. Because $d(x,S) = 0$, the condition $\delta > \frac{1}{4}d(x,S)$ is trivially satisfied, placing us inherently in the ``Large distance'' regime where everything holds under the worst-case bound without a $\log k$ penalty. Finally, if both $x, y \in S$, their distance is determined entirely by the embedding drawn from $\cD_S$, so it is bounded by $c_S\cdot d(x,y)$ in expectation.
\end{proofof}

%% file: proofs.tex
\section{Missing proofs from \texorpdfstring{\Cref{sec:merge}}{Section \ref*{sec:merge}} } \label{sec:mergeproofs}

\subsection{Proof of \texorpdfstring{\Cref{thm:bilipl1}}{Theorem \ref*{thm:bilip-hst}} }\label{sec:bilipl1-lastproof}

\begin{proof}[Proof of \Cref{thm:bilipl1}, continued]
Here we prove the second part of \Cref{thm:bilipl1}.

First, we assume that vectors in the image of $\alpha_S$ have dimension $\tau$. Our final embedding $\alpha$ will have $\tau+n^2$ coordinates, and for each $x\in S$, $\alpha(x)$ will match $\alpha_S(x)$ on the first $\tau$ indices and be $0$ on all other indices.
We will define two embeddings on $X$: $\alpha_c$, which we call the ``core embedding,'' and $\alpha_e$, which we call the ``extra embedding.'' To ensure we meet the requirements of strong \bl extensions, we will require that for all $x\in S$, $\alpha_c(x)=\alpha_S(x)$ and $\alpha_e(x)=0^{n^2}$ (the $0$ vector in $n^2$ dimensions). We define each of these embeddings as follows. Let $\cD$ be the probabilistic embedding \Cref{alg:extend} samples from when the input is $\alpha_S$. 

\begin{itemize}
    \item Let $\cD_c$ be the probabilistic embedding found by sampling from $\cD$ and truncating all points in the image of the embedding by restricting them to the first $\tau$ indices (the ones corresponding to $\alpha_S$).  
    

    Let $\alpha_c(x)$ be the coordinate-wise average of the embeddings of $x$ under $\cD_c$. That is, if $\cD_c$ outputs $\alpha^t$ with probability $p_t$, then $\alpha_c(x):=\sum_{t}p_t\cdot \alpha^t(x)$, where the sum is a coordinate-wise sum over vectors.

    Note that for all $x\in S$, $\alpha_c(x)=\alpha_S(x)$, as the first $\tau$ indices of $\alpha(x)$ match $\alpha_S(x)$ for all embeddings $\alpha$ in the support of $\cD$. 
    
    \item Let $\cD_e$ be the probabilistic embedding found by sampling from $\cD$ and deleting the first $\tau$ coordinates of every vector in the image of the embedding. As discussed earlier, there must be some deterministic embedding $\alpha_e$ such that for all $x,y\in X$, $d_{\alpha_e}(x,y)=\expec_{\alpha\sim \cD_e}[d_\alpha(x,y)]$. 

    Note that any set of $n$ points in finite $\ell_1$ space can be isometrically embedded into $n^2$-dimensional $\ell_1$ space \cite{ball1990isometric}, so we assume wlog that $\alpha_e$ maps points to vectors in $n^2$-dimensional space. Further, note that for any $x,y\in S$, an embedding $\alpha\sim \cD$ can only have $\alpha(x)$ and $\alpha(y)$ differ on the first $\tau$ indices (as they are mapped to their vectors under $\alpha_S$ on the first $\tau$ indices, and they are all mapped to the same vector as the ``anchor point'' of subsequent cluster embeddings $\alpha_i$). Thus, their distance is $0$ under $\alpha_e$ and they are all mapped to the same point. We assume wlog (by translation) that this point is the all $0$s vector.
\end{itemize}

Now we are ready to define our final embedding $\alpha$. For all $x\in X$, we let $\alpha(x)=\alpha_c(x)\circ \alpha_e(x)$. Note that if $x\in S$, $\alpha_c(x)=\alpha_S(x)$ and $\alpha_e(x)=0^{n^2}$, so if we project $\alpha(x)$ to the subspace spanned by $\set{\alpha(y):y\in S}$, we get exactly $\alpha_S(x)$.
Thus, we are left with bounding expansion and contraction for general pairs of nodes in $X$.

\textbf{Expansion:}
Consider a pair of points $x,y\in X$. Let $p_t$ be the probability that $\alpha^t$ is sampled from $\cD_c$. (Note that the support of $\cD_c$ is finite.) We have

\begin{align*}
    d_{\alpha_c}(x,y) &= \sum_{i=1}^\tau |\alpha_c(x)_i-\alpha_c(y)_i|\\
    &= \sum_{i=1}^\tau |\sum_{t}p_t(\alpha^t(x)_i-\alpha^t(y)_i)| \\
    &\leq\sum_{t} p_t\sum_{i=1}^\tau|\alpha^t(x)-\alpha^t(y)| \\
    &= \sum_t p_t \cdot d_{\alpha^t}(x,y)\\
    &= \expec_{\alpha'\sim \cD_c}[d_{\alpha'}(x,y)].
\end{align*}

Thus we get

\begin{align*}
    d_\alpha(x,y) &= \alpha_c(x,y) + \alpha_e(x,y) \\
    &\leq \expec_{\alpha_1\sim \cD_c}[d_{\alpha_1}(x,y)] + \expec_{\alpha_2\sim \cD_e}[d_{\alpha_2}(x,y)] \\
    &\leq \expec_{\alpha'\sim \cD}[d_{\alpha'}(x,y)] \\
    &\leq O(\log k\cdot c_S)\cdot d(x,y).
\end{align*}

\textbf{Contraction:}
Now we bound contraction for each pair of nodes $x,y\in X$. We begin by considering $x,y\in K$, and we divide into two cases. 

    \begin{itemize}
        \item Case 1: $d(x,S)+d(y,S)\geq  \frac{d(x,y)}{7}$

        In this case, we show that the distance between $x$ and $y$ is at least $d(x,y)/7$ under $\alpha$. In particular, we show $d_{\alpha_e}(x,y)\geq d(x,y)/7$. 

        As $d_{\alpha_e}(x,y)=\expec_{\alpha'\sim \cD_e}[d_{\alpha'}(x,y)]$, it will be sufficient for us to show that $d_{\alpha'}(x,y)\geq d(x,y)/7$ for all $\alpha'$ in the support of $\cD_e$. 

        Let $\alpha'$ be in the support of $\cD_e$, and let $(\cK,r)$ be the partition that was sampled in the process of sampling $\alpha'$ from $\cD_e$. If $x,y\in K_i$ for some $i$, then $d_{\alpha_i}(x,y)\geq d(x,y)$ because FRT embeddings are non-contracting. Thus, $d_{\alpha'}(x,y)\geq d(x,y)$, by the criteria of perfect merge functions. 

        Otherwise, let $x\in K_i,y\in K_j$ for $i\neq j$, and let $r(x)$ and $r(y)$ be the anchor points of $x$ and $y$. We have $d_{\alpha_i}(x,r(x))\geq d(x,r(x))$ and $d_{\alpha_j}(y,r(y))\geq d(y,r(y))$ by the fact that FRT embeddings are expanding. 

        Note that because our merge function is just concatenation, under $\alpha'$, $x$ and $y$ differ on exactly the indices corresponding to $\alpha_i$ where $y$ is mapped to the same vector as $r(x)$ and those corresponding to $\alpha_j$, where $x$ is mapped to the same vector as $r(y)$. Thus, their distance under $\alpha'$ is $d_{\alpha_i}(x,r(x))+d_{\alpha_j}(y,r(y))\geq d(x,r(x))+d(y,r(y))\geq d(x,S)+d(y,S)\geq d(x,y)/7$.


        \item Case 2: $d(x,S)+d(y,S)<\frac{d(x,y)}{7}$

        In this case, we show that $x$ and $y$ are sufficiently far apart on the coordinates corresponding to $\alpha_c$. In particular, let $S_x$ be the set of points in $S$ that could be assigned as $x$'s anchor point, and let $S_y$ be the same for $y$.    
        The coordinate-wise average map for $x$  under $\cD$ must be in the convex hull of $S_x$, and likewise the point $y$ is mapped to must be in the convex hull of $S_y$. We will argue that points in the convex hull of $S_x$ and $S_y$ must be far apart.

        Let $s$ be the point in the convex hull of $S_x$ that $x$ is mapped to. Likewise, let $t$ be the point in the convex hull of $S_y$ that $y$ is mapped to. Then we are interested in $d(s,t)$. Let $a$ be any point in $S_x$ and let $b$ be any point in $S_y$. Note that $d(a,s)$ is upper bounded by the maximum distance between any pair of points in $S_x$. However, if $a,a'\in S_x$, then $d(x,a)\leq 2d(x,S)$ and $d(x,a')\leq 2d(x,S)$, so $d(a,a')\leq d(a,x)+d(x,a')\leq 4d(x,S)$. Thus, $d(s,a)\leq 4d(x,S)$. Likewise, $d(b,y)\leq 4d(y,S)$. 
        
        Additionally, we argue that $a$ and $b$ are ``far apart.'' We have $d(a,x)\leq 2d(x,S)$ and $d(b,y)\leq 2(y,S)$. Thus, by the triangle inequality $d(a,b)\geq d(x,y)-d(a,x)-d(b,y)\geq d(x,y)-2(d(x,S)+d(y,S))$.
        
        We get 

        \begin{align*}
            d(s,t) &\geq d(a,b) - d(a,s) - d(b,t) \\
            &\geq d(a,b) - 4(d(x,S)+d(y,S)) \\
            &\geq d(x,y)-6(d(x,S)+d(y,S))\\
            &> d(x,y)-\frac{6}{7}d(x,y) \\
            &= \frac{d(x,y)}{7}.
        \end{align*}

\end{itemize}

If $x,y\in S$, we have already lower bounded their distance by $1$. If $x\in K,y\in S$, we can consider a point in $K$ a distance $0$ from $y$ and note that the above analysis applies. This point must be mapped to the same point as $y$, so this implies our desired result.
\end{proof}

\subsection{Proof of \texorpdfstring{\Cref{lem:merge-hst}}{Lemma \ref*{lem:merge-hst}} }\label{sec:hstmergeproof}

In this section we prove that \Cref{alg:merge} is a perfect merge function with contraction factor $2$.

\begin{proofof}{\Cref{lem:merge-hst}}
Let $x\in Z_1$ and $y\in Z_2$ with $Z_1\cap Z_2=\set{u}$ and let $\alpha_1:Z_1\rightarrow T_1$, $\alpha_2:Z_2\rightarrow T_2$, where both $T_1$ and $T_2$ are $\beta$-HSTs with elements of $Z_1$ and $Z_2$ mapped to their leaves respectively. Let $\alpha\from \merge(\alpha_1,\alpha_2)$ be the resulting $\beta$-HST embedding from $Z_1\cup Z_2$ into $T$. 

We first argue that $T$ is indeed a $\beta$-HST. 
Consider a node at height $i$ in $T$. If this node was originally in $T_1$, then any subtree that may have been added as a subtree of this node has height exactly $i-1$ (by the fact that they are copies of subtrees of the children of $\alpha_2(u)$'s ancestor at height $i$), so this node is still at height $i$ and is still labeled properly. Now consider the case that this is a copy of a node from $T_2$. If the original was at height $i$ in $T_2$, then this copy is also at height $i$, as its subtree is identical to that in $T_2$. Therefore, the node still retains the correct label. 

We will show each of the \goodmerge\ merge criteria separately. Let $LCA(x,y)$ be the least common ancestor of $\alpha(x)$ and $\alpha(y)$ in the final merged tree $T$.

\begin{itemize}

    \item First consider the case that $x,y\in Z_1$. Note that $T$ is identical to $T_1$ except with some extra children added to some internal nodes, and $\alpha$ maps $x$ and $y$ to exactly the same leaves as under $\alpha_1$, so their distance is unchanged.

    \item Now consider $x,y\in Z_2$.  Let $i\st$ be the height of the least common ancestor of $\alpha_2(x)$ and $\alpha_2(y)$ in $T_2$. Let $S_i(x)$ be the subtree of $T_2$ of height $i$ that contains $x$. Note that by definition of $i\st$, $S_{i\st-1}(x)\neq S_{i\st-1}(y)$ (which also holds for all $i<i\st$). 

    First, consider the case that $S_{i\st-1}(x)\neq S_{i\st-1}(u)$ and  $S_{i\st-1}(y)\neq S_{i\st-1}(u)$. Then the root of $S_{i\st-1}(x)$ and root of $S_{i\st-1}(y)$ are both made children of $\alpha(u)$'s ancestor at height $i\st$. Thus, their least common ancestor is still at level $i\st$ and their distance is unchanged. 

    Next, consider the case that $S_{i\st-1}(x)=S_{i\st-1}(u)$ and $S_{i\st-1}(y)\neq S_{i\st-1}(u)$. (Note that at most one can be equal to $S_{i\st-1}(u)$ so we may assume this without loss of generality.) The first ancestor $\alpha(x)$ and $\alpha(y)$ could have in common is $\alpha(u)$'s ancestor at level $i\st$, as $\alpha(x)$'s first ancestor from $T_1$ is below this level (and thus all ancestors above this level are also from $T_1$), and $\alpha(y)$'s lowest ancestor in common with $u$ is at level $i\st$. Further, they do have this ancestor in common by construction, so the least common ancestor height under $\alpha$ is the same as under $\alpha_2$. Thus, we get $\dista(x,y)=\dista[\alpha_2](x,y)$.

    \item Finally consider the case that $x\in Z_1\setminus\set{u}$ and $y\in Z_2\setminus\set{u}$. Let $LCA(a,b)$ be the least common ancestor of $\alpha(a)$ and $\alpha(b)$ in $T$. Note that by definition of HSTs, for any set of nodes $\set{a,b,c,d}$, we have $\dista(a,b)\geq \dista(c,d)$ if and only if $LCA(a,b)$ is at or above the height of $LCA(c,d)$. 
    
    {We will show $\dista(x,y)\geq \dista(x,u)= \dista[\alpha_1](x,u)$ and $\dista(x,y)\geq \dista(y,u)= \dista[\alpha_2](y,u)$ separately.} 

    \begin{itemize}
        \item $\dista(x,y)\geq \dista(x,u)$: Note that $LCA(x,y)$ is an  original member of $T_1$, as nodes in $T_1$ only have ancestors from $T_1$. Further, by construction all ancestors of $\alpha(y)$ that are from $T_1$ are also ancestors of $\alpha(u)$. Thus, $LCA(x,y)$ must also be an ancestor of $\alpha(u)$ and it must be at or above the height of $LCA(x,u)$. 

        \item $\dista(x,y)\geq \dista(y,u)$: Note that $LCA(y,u)$ is $\alpha(y)$'s lowest height ancestor that was an original member of $T_1$, as all ancestors below that height are from a copied subtree from $T_2$. Further, since $\alpha(x)$ only has ancestors that were original members of $T_1$, $LCA(x,y)$ must be at or above the height of $LCA(y,u)$.  
    \end{itemize}
    Thus, $d_\alpha(x,y)\geq \max\set{d_\alpha(x,u),d_\alpha(u,y)}\ge \frac{d(x,u)+d(u,y)}{2}$. 
\end{itemize}
\end{proofof}

%% file: proofs_outlier.tex
\section{Omitted proofs from Section \ref{sec:hst}\label{sec:proofoutlier}}

In this appendix, we formally show that the rounding algorithm we present for \ref{eqn:hst1} succeeds in giving the promised approximation from Theorem \ref{thm:outlier}. We also formalize the algorithm described in Section \ref{sec:hst}, and for completeness we write out the algorithms of Munagala et al. \cite{mungala2023} used in this paper.

\subsection{Bounding the optimal LP solution}

To bound the size of the final outlier set obtained by our rounding algorithm, we first need to argue that the optimal solution to this LP has value at most $k$. This is expressed in Lemma \ref{claim:lp-valid}.

\begin{lemma}
\label{claim:lp-valid}
If \ref{eqn:hst1} is parameterized by $((X,d),k,c)$ and $(X,d)$ has a probabilistic $(k,c)$-outlier embedding into HSTs, then the optimal solution to \ref{eqn:hst1} has value at most $k$.
\end{lemma}

This lemma is true for the same reasons given in \cite{mungala2023}.\footnote{Note that in order to obtain Lemma \ref{claim:lp-valid}, we need to have that the final embedding has no contraction, even on outliers, as we need to be able to guarantee that any node $j$ is represented at level $r$ by a node $i$ in $B_j(r)$. If we allow contraction for outliers, this may not be the case, and $j$ may be represented at each level by something that is actually quite far away from it. } The primary change we have made is that $c$ is now a constant, and in finding a feasible solution to the LP, we need to assign values to the $\delta_i$ in addition to assigning the other variables as in \cite{mungala2023}. Note that if we assign $\delta_i$ to be $1$ if $i$ is an outlier, then the right hand side of constraint \ref{eqn:hst2} is an upper bound on the distortion for any pair of nodes under the nested composition of embeddings by \Cref{thm:bilip-hst}. 
Note in particular that constraints \ref{eqn:hst4} through \ref{eqn:hst7} require that the underlying embedding be {\em expanding} on all pairs, even when one or both nodes may be an outlier. Thus, bounding the optimal solution to this LP is why it was essential to find nested compositions of embeddings for HSTs, rather than just Lipschitz extensions.

\begin{proofof}{Lemma \ref{claim:lp-valid}}
Let $\cD_S$ be a probabilistic embedding on some $S\subseteq X$ into HSTs  with distortion $c$ such that $K:=|X\setminus S|$ has size at most $k$. By \Cref{thm:bilip-hst} (plus scaling up by a factor of $4$), there exists a probabilistic embedding $\cD$ that is non-contracting, has distortion at most $4c$ on pairs of nodes in $S$, and has overall distortion at most $\zeta\cdot (c\log k)$ for some constant $\zeta$. Set the variables $z_{ijj'}^r,\gamma_{jj'}^r,$ $x_{ij}^r$, and $\delta_i$ as in Table \ref{tab:variables}. (That is, $z_{ijj'}^r$ is the probability that $i$ is the ``representative'' of $j$ and $j'$ at level $r$, $\gamma_{jj'}^r$ is the probability that $j$ and $j'$ are separated at level $r$, $x_{ij}^r$ is the probability that $i$ is the representative of $j$ at level $r$, and $\delta_i$ is $1$ if $i$ is an outlier and $0$ otherwise.) By how we set the $\delta_i$ values, this solution clearly has value $k$, so we need only show it satisfies the constraints.

Each of the constraints in this LP is satisfied for essentially the same reason as in \cite{mungala2023}. We will briefly state those reasons here, paying special attention to Constraint \ref{eqn:hst2}, as we did alter this constraint from \cite{mungala2023} (though it is satisfied for essentially the same reasons as in \cite{mungala2023}).

To see why constraint \ref{eqn:hst2} is satisfied, let $\mu_{jj'}^r$ be the probability that $j$ and $j'$ are a distance exactly $r$ apart in the embedding sampled from $\cD$. Note that if $j$ and $j'$ are separated at level $r$, then their distance in that embedding is {\em at least} $r\st>r$. Thus, $\gamma_{jj'}^r$ is the probability that the embedding drawn from the distribution places $j$ and $j'$ at a distance strictly greater than $r$ apart, which means $\gamma_{jj'}^r=\sum_{r\st\in M:r\st>r}\mu_{jj'}^{r\st}$. Thus we get

\begin{align*}
\sum_{r\in M}r\cdot \gamma_{jj'}^{r} &= \sum_{r \in M}r\cdot \sum_{r\st\in M:r\st>r}\mu_{jj'}^{r\st}\\
&= \sum_{r\st\in M}\mu_{jj'}^{r\st}\sum_{r\in M: r\st>r}r \\
&\leq \sum_{r\st \in M}r\st\mu_{jj'}^{r\st} \\
&= \expec_{\alpha\sim \cD}\left[\dista(j,j') \right].
\end{align*}

The first line is by what we argued above; the second is by reordering the sums; the third is by the fact that $\sum_{r\in M:r<r\st}r\leq r\st$ (as all values in $M$ are powers of $2$); the fourth line is by definition of $\mu_{jj'}^r$ as the probability that $j$ and $j'$ are a distance exactly $r$ apart. Note that if at least one of $j,j'$ is an outlier, then  at least one of $\delta_j$ and $\delta_{j'}$ is $1$, so the right hand side of constraint \ref{eqn:hst2} is at least $\zeta\cdot c\cdot \log k\cdot d(j,j')$, which is an upper bound on the expected distance between $j$ and $j'$ under $\cD$. If neither $j$ nor $j'$ is an outlier, then the right hand side of constraint \ref{eqn:hst2} is $4c\cdot d(j,j')$, which is an upper bound on their expected distance under $\cD$.

Now we can briefly describe why the other constraints are satisfied.

\begin{enumerate}
    \item Constraint \ref{eqn:hst3} is satisfied, as the probability that both $j$ and $j'$ have $i$ as a representative at level $r$ is at most the probability that either of them individually has this representative.

    \item 
    Constraint \ref{eqn:hst4} is satisfied because the embedding is non-contracting, so every node has a representative within a distance $r$ of it at level $r$.

    \item 
    Constraint \ref{eqn:hst5} is
    because by non-contraction of the embedding, we have $x_{ij}^r=0$ if $j\notin B_i(r)$ so the left side is the probability that $j$ has some representative at level $r$, which happens with probability $1$.

    \item Constraint \ref{eqn:hst7} is satisfied by non-contraction of all embeddings, and constraints \ref{eqn:hst8} and \ref{eqn:hst9} are satisfied because all probabilities are in the range $[0,1]$. 
                           
\end{enumerate}
\end{proofof}

\subsection{Rounding the LP \label{sec:lproundproof}}

We formally describe the rounding algorithm of Section \ref{sec:hst} in Algorithm \ref{alg:round}. This algorithm largely relies on the algorithms of \cite{mungala2023} (presented here as Algorithms \ref{alg:alg1hst} and \ref{alg:alg2hst}), but it takes the extra step of excluding from the domain of the output embedding all nodes in $X$ with optimal $\delta_j$ value above some threshold $\delta\st$. These excluded nodes form the outlier set for this embedding. Note that the outlier set is deterministic, as it is decided by a simple thresholding of the deterministic LP solution.

\input{alg_round}

\input{alg1hst}

\input{alg2hst}

In order to show Algorithm \ref{alg:round} does in fact sample from a probabilistic $(O(\frac{k}{\epsilon}\log^2k),(32+\epsilon)c)$-outlier embedding, we will reference the following lemmas from \cite{mungala2023}.

\begin{lemma}[counterpart of Lemma 10 from \cite{mungala2023}]\label{lem:mstlem10}
The embedding produced by the rounding Algorithms \ref{alg:alg1hst} and \ref{alg:alg2hst} is non-contracting. 
\end{lemma}

\begin{lemma}[counterpart of Lemma 11 from \cite{mungala2023}]\label{lem:mstlem11}
If $\delta_j$ is the optimal LP value of the variable $\delta_j$ and if $\alpha$ is sampled by Algorithms \ref{alg:alg1hst} and \ref{alg:alg2hst}, we have $\expec_\alpha[\dista(j,j')]\leq 8\cdot (4+\zeta\cdot \log k\cdot (\delta_j+\delta_{j'}))\cdot c\cdot d(j,j')$
\end{lemma}

Note that Lemma \ref{lem:mstlem10} can be proven identically to \cite{mungala2023}, as the relevant algorithms and constraints involved in that proof are unchanged, but we briefly outline their reasoning here.

\begin{proofof}{Lemma \ref{lem:mstlem10}}
As Algorithms \ref{alg:alg1hst} and \ref{alg:alg2hst} are unchanged from \cite{mungala2023}, we can use the same reasoning. In particular, for any pair of nodes $j,j'$, if $j$ and $j'$ are in the same part of the partition at the level corresponding to $r$, then there is some $i$ that is a distance at most $r$ from each of them. This means that the first level at which they share a parent is the level corresponding to $r$, and since we double all distances at the end of Algorithm \ref{alg:alg2hst} by scaling everything up by $2$, this implies non-contraction.
\end{proofof}

\begin{proofof}{Lemma \ref{lem:mstlem11}}
Lemma 11 of \cite{mungala2023} stated that $\expec_\alpha[\dista(j,j')]\leq 8q$, where $q$ is the optimal LP solution value for their distortion variable. (Recall that for us, the distortion variable is replaced with a constant and we focus on our new outlier variables.) To prove this, they used only constraints \ref{eqn:hst3} through \ref{eqn:hst9} to show that $\expec_\alpha[\dista(j,j')]\leq 8\cdot \sum_{r\in M}r\cdot \gamma_{jj'}^r$ and then applied their version of constraint \ref{eqn:hst2} to get the final result. As constraints \ref{eqn:hst3} through \ref{eqn:hst9}  and the rounding algorithm remain unchanged in our algorithm, the statement $\expec_\alpha[\dista(j,j')]\leq 8\cdot \sum_{r\in M}r\cdot \gamma_{jj'}^r$ still holds. Thus, by our slightly modified version of constraint \ref{eqn:hst2}, we  get $\expec_\alpha[\dista(j,j')]\leq 8\cdot \sum_{r\in M}r\cdot \gamma_{jj'}^r\leq 8\cdot (4+\zeta\cdot \log k\cdot (\delta_j+\delta_{j'}))\cdot c\cdot d(j,j')$, as stated in the lemma. 
\end{proofof}

Finally we can prove Theorem \ref{thm:outlier}, which says that we can efficiently sample from an $(O(\frac{k}{\epsilon}\log k),(32+\epsilon)c)$ embedding. 

\begin{proofof}{Theorem \ref{thm:outlier}}
We know from Lemma \ref{lem:mstlem10} that the sampled embedding from Algorithm \ref{alg:round} is non-contracting. To consider its expected expansion, we need only consider pairs of nodes $j,j'$ with $\delta_j,\delta_{j'}\leq \delta\st$. Then by Lemma \ref{lem:mstlem11}, the expected distance between such a $j$ and $j'$ is at most 

\begin{align*}
    (32+8\zeta\cdot \log k\cdot 2\delta\st)\cdot c\cdot d(j,j') \leq (32+\epsilon)\cdot c\cdot d(j,j'),
\end{align*}

by the choice of $\delta\st$ in Algorithm \ref{alg:round}. 

Additionally, note that the set of outliers produced by Algorithm \ref{alg:round} is deterministic and has size at most $k/\delta\st=O(\frac{k}{\epsilon}\log k)$, as Lemma \ref{claim:lp-valid} tells us that the optimal solution to \ref{eqn:hst1} has value at most $k$.
\end{proofof}

\subsubsection{A note on the runtime}

Notably, Algorithm \ref{alg:alg1hst} (Algorithm 1 from \cite{mungala2023}) runs in {\em expected} polynomial time, not deterministic polynomial time. Following the same reasoning as \cite{mungala2023}, we can select some additional error term $\epsilon'$ and a sufficiently large constant $c'$ and terminate the while loop after at most $c'\cdot n\ln\frac{\Delta}{\epsilon'}$ steps. If the algorithm would have otherwise exceeded that number of steps, we just output an HST in which all distances are $2^{\lceil \log\Delta\rceil}$. Using sufficient concentration bounds, the rounding scheme incurs a distortion at most $\epsilon'$ higher than in the original form. If we want the distortion of the final embedding to be at most $(32+\epsilon)c$, we can select $\epsilon/2$ for $\epsilon'$ in this section and for the approximation factor in choosing a threshold $\delta\st$. Note that because of this choice, the runtime of the algorithm depends on $1/\epsilon$ and $\Delta$ (specifically growing with $\ln(\frac{\Delta}{\epsilon})$), unlike the $\ell_2$ outlier embedding algorithm in \cite{chawla2024}.

\subsection{Weighted outlier embeddings \label{sec:weighted}}

In the weighted outlier set problem, we are given $\cX=(X,d)$ and $c$, as well as a set of weights $w$ on the nodes, and our goal is to find a subset $K\subseteq X$ such that $X\setminus K \probembed{c}\cH$ and $w(K)$ is minimized.
So far we have considered the unweighted version of this problem, in which $w(i)=1$ for all $i\in X$. In this subsection, we will discuss the slight changes needed to apply the results of Section \ref{sec:hst} to the weighted outlier setting and thus obtain a bicriteria approximation for this problem.

The weights of the nodes affect the cost of leaving those nodes out, but they have no effect on how good the resulting embedding may be. This means that if we attempt to modify \ref{eqn:hst1} for the weighted setting, the weights should only influence the objective function.  
To solve this problem, we will use Algorithm \ref{alg:round} with the changes noted in red for weighted outlier embeddings.
In the weighted setting, we are not necessarily concerned with simply minimizing the outlier set size, so intuitively we should select the LP solution corresponding to the $k$ that minimizes the objective value function for that LP. However, in order to account for the extra $\log k$ approximation factor incurred by our algorithm, we will instead look for a value of $k$ that minimizes the objective value of its LP solution, times $\log k$. The necessary changes are noted as  comments in Algorithm \ref{alg:round}. Below we give the proof of Corollary \ref{thm:weighted}.

\begin{proofof}{Corollary \ref{thm:weighted}}
Let $K$ be an outlier set of $\cX$ that has size at most $\sminoutset{c}$, has weight at most $\sminoutval{c}$, and has $X\setminus K\probembed{c}\cH$. Define $k:=|K|$.

Let $k\st$ be the value chosen by the algorithm to minimize $v_{k'}\log{k'}$. Then
\begin{align*}
    v_{k\st}\cdot \log(k\st) &\leq v_k\cdot \log k \\
    &\leq \sminoutval{c} \cdot \log \sminoutset{c},
\end{align*}

where the first line is by definition of $k\st$ and the second line is by the fact that $\delta_i=\begin{cases}
    1 & i\in K \\
    0 & else
\end{cases}$ is part of a valid solution to \ref{eqn:hst1} with parameter $k$, and this solution has objective value at most $\sminoutval{c}$. (Note that this is part of a valid solution because only the objective function has changed, so all arguments with respect to feasibility of such a solution still go  through.) 

Now we consider the rounding step of the algorithm. We know that by selecting all $i$ with $\delta_i\geq \frac{\epsilon}{\zeta\cdot \log k\st}$ in the LP solution corresponding to $k\st$, we obtain a solution with at most $(32+\epsilon)c$ distortion on the non-outliers (again following identical reasoning to that in the unweighted case). Thus, we are just left with bounding the weight of the outlier set. We have that the total cost is at most $v_{k\st}\cdot \frac{\zeta \cdot \log(k\st)}{\epsilon}$ as a result of our rounding, as $v_{k\st}$ is the value of the objective function under this solution. This implies an outlier set of weight at most $O(\frac{1}{\epsilon}\cdot \sminoutval{c} \cdot \log \sminoutset{c})$.
\end{proofof}

%% file: alg_round.tex
\begin{algorithm}
\caption{Algorithm for approximating outlier sets for HST embeddings}
\label{alg:round}
\textbf{Input:} Metric space $(X,d)$, target distortion $c$, approximation factor $\epsilon>0$ \\
\textbf{Output:} An embedding $\alpha:S\rightarrow T$ for some HST $T$ and some $S\subseteq X$
\begin{algorithmic}[1]
\For{$k=1$ to $n$}
\State Let $LP_k$ be \ref{eqn:hst1} for $(X,d),c,k$ \Comment{{\scriptsize\change{For weighted outlier sets, use  $\min \sum_{i\in X}w_i\delta_i$ as the objective}}}
\State Find an optimal solution to $LP_k$ of value $v_k$
\State Use the rounding algorithm of \cite{mungala2023} (Algorithms \ref{alg:alg1hst} and \ref{alg:alg2hst}) to sample an embedding $\alpha_k:X\rightarrow T$ for HST $T$
\State $\delta\st\leftarrow \frac{\epsilon}{16\zeta\cdot \log k}$
\State $K_k\leftarrow \set{j\in X:\delta_j\geq \delta\st\text{ in the optimal LP solution}}$
\EndFor
\State Let $k\st$ be the smallest  value of $k$ with $v_k\leq k$ \Comment{{\scriptsize\change{For weighted outlier sets, let $k\st$ minimize $v_k\log k$}}}
\State Return $\alpha_{k\st}$ and $K_{k\st}$
\end{algorithmic}
\end{algorithm}

%% file: alg1hst.tex
\begin{algorithm}
\caption{(Algorithm 1 from \cite{mungala2023}): Random selection of partitions for the nodes in $X$}
\label{alg:alg1hst}
\begin{algorithmic}[1]
\For{$r\in M$ in decreasing order}
\State $S\leftarrow X$; $P_i^r\leftarrow \emptyset$ for all $i\in  X$
\While{$S\neq \emptyset$}
\State Choose a center $i\in  X$ uniformly at random independent of past choices 
\State Choose $\ell_i^r\in [0,1]$ uniformly at random independently of past choices
\State For each $j\in S\cap B_i(r)$, if $x_{ij}^r\geq \ell_i^r$, assign $j$ to $P_i^r$ and remove $j$ from $S$.
\EndWhile
\EndFor
\end{algorithmic}
\end{algorithm}

%% file: alg2hst.tex
\begin{algorithm}
\caption{(Algorithm 2 from \cite{mungala2023}): Embedding $(X,d)$ (with diameter $\Delta$) into an HST given randomly chosen partitions $P_i^r$ from Algorithm \ref{alg:alg1hst}}
\label{alg:alg2hst}
\begin{algorithmic}[1]
\State Place a root node $w$ at the highest level with $\levelval_w\leftarrow2\cdot 2^{\lfloor \log \Delta \rfloor}$ and set $S_w\leftarrow X$
\For{$r\in M$ in decreasing order}
\For{each node $w$ at the previous (parent) level with set $S_w$}
\For{each $i$ with $P_i^r\cap S_w\neq \emptyset$}
\State Place a child node $v$ with set $S_v\leftarrow P_i^r\cap S_w$ and $\levelval_v\leftarrow 2r$
\EndFor
\EndFor
\EndFor
\end{algorithmic}
\end{algorithm}

%% file: extra_applications.tex
\section{More applications\label{sec:extra-app}}

In this appendix, we discuss some interesting applications of probabilistic outlier embeddings into HSTs. In particular, we start by elaborating on our discussion of the MCCT problem from Section \ref{sec:applications}. Then we discuss the existence of metrics with small outlier sets for probabilistic HST embeddings and two other problems (buy-at-bulk and dial-a-ride) whose best-known solutions can be improved for certain inputs using the techniques from this paper.


\subsection{MCCT with outliers}


In Section \ref{sec:applications}, we argued that we can use the outlier embeddings algorithm from this paper to obtain a solution with at most $O((1+\epsilon)c_k(\cX)(OPT_k))$ cost and $O(\frac{k}{\epsilon}\log k)$ outliers, where $c_k(\cX)$ is the minimum distortion for embedding $\cX$ into HSTs with at most $k$ outliers and $OPT_k$ is the cost of the optimal solution that is allowed to select $k$ outliers. We will discuss this in additional detail here.

In particular, assume that an optimal algorithm for the $k$-outlier MCCT problem removes a subset $K\st\subseteq X$. Let $S=X\setminus K\st$. The cost of this algorithm, which is $OPT_k$, is at least $\sum_{i,j\in S}r_{ij}d(i,j)$. 
Let $OPT_{\alpha,k}$ be the cost of the optimal $k$-outlier solution on the metric $d_\alpha$ (i.e. the optimal solution after the embedding $\alpha$ is applied).
We get that if $\cD$ is a probabilistic embedding with expected distortion at most $c$, then $\expec_{\alpha\sim \cD}[OPT_{\alpha,k}]\leq E_\alpha[\sum_{i,j\in S}r_{ij}d_\alpha(i,j)]\leq c\sum_{i,j\in S}r_{ij}d(i,j)\leq 8c\cdot OPT_k$, where we used the fact that the optimal cost for the $k$-outlier problem on $d_\alpha$ is bounded by the cost of embedding all points but $K\st$ into a tree, and we applied the Steiner point removals results of Gupta \cite{gupta2001steiner}. Consider probabilistically embedding a given metric into trees using $\cD$ and then applying the outlier approximation algorithm from Claim  \ref{claim:mcct-out}. For an input metric $d_\alpha$, this algorithm produces a solution with at most $3k$ outliers and cost at most $3OPT_{\alpha,k}$. Thus, the expected cost of this solution is at most $\expec_{\alpha\sim \cD}[3OPT_{\alpha,k}]\leq 24c\cdot OPT_k$, where $c$ is the distortion of embedding $\cD$. This is enough to obtain the results we claimed in Section \ref{sec:applications}.

The outlier embedding results from this paper let us improve the factor $c$ that we can obtain for such an approximation if we are willing to increase the number of outliers. Note also that because $\epsilon$ is a tunable parameter in our algorithm, it does not necessarily need to be a constant, so if we are more interested in minimizing the number of outliers, we can also obtain a solution of cost at most $O(\log k \cdot c_k(\cX)\cdot OPT_k)$ with only $O(k)$ outliers.

\input{bad_graphs}

\input{buyatbulk}

\input{dial}

%% file: bad_graphs.tex
\subsection{Some metrics have small outlier sets}
To motivate this work, we may ask: do there exist  metrics for which removing a small number of points results in a significantly better probabilistic embedding into HSTs?  In this subsection, we discuss types of metrics for which this is the case. In particular, consider the distance metric on an unweighted graph obtained by ``linking'' two unweighted graphs through an edge, the first graph being an expander on $n$ nodes and the second being the complete graph on $n-\log n$ nodes.  Bartal \cite{bartal1996} shows that, for any $x$, certain expander graphs on $x$ nodes incur a distortion $\Omega(\log x)$ when probabilistically embedding into HSTs. Thus, the graph we have described requires $\Omega(\log \log n)$ distortion  for HST embeddings (inheriting the distortion bound of the expander subgraph), whereas the same graph with the expander removed can be embedded with $O(1)$ distortion.

This graph is highly specific and the ``good'' and ``bad'' parts of the graph are easy to spot. We will use this graph as an example in the coming subsections because its simplicity makes it easy to analyze, but we also note that the class of metrics whose optimal distortion can be improved by removing a small number of vertices is much more diverse than this simple example.

To demonstrate a large family of metrics with good outlier embeddings, we will turn to the notion of {\em metric composition} as defined by Bartal et al. \cite{bartal2003}. 

\begin{definition}[Metric composition \cite{bartal2003}]
Let $(M,d_M)$ be a metric space with minimum distance exactly $1$, and let $\cN=\set{(N_x,d_{N_x})}_{x\in M}$ be a set of metric spaces. Let $\Delta := \max_{x\in M}\set{diameter(N_x)}$. Then the $\beta$-composition of $M$ and $\set{N_x}$ denoted $M\st=M_\beta[\cN]$ is a metric on the vertex set $\set{(x,u):x\in M,u\in N_x}$ such that distances in $M\st$ are
\begin{align*}
    d_{M\st}((x,u),(y,v)) := \begin{cases}
        d_{N_x}(u,v) & x=y \\
        \beta\cdot \Delta\cdot  d_M(x,y) & else
    \end{cases}.
\end{align*} 
\end{definition}

We provide a visualization of this composition in Figure \ref{fig:compose}. Bartal et al. 
\cite{bartal2003} show that for any $M$ and $\cN$, $M_\beta[\cN]$ defines a proper metric if $\beta\geq 1/2$. While it's easy to factor a composition back into its original components if $\beta$ is quite large, it's less clear how to do so when $\beta=1/2$.

\begin{figure}
    \centering
    \input{figures/compose}
    \caption{A visual example of metric composition where $M$ is some metric on the set $\set{x,y,z}$.}
    \label{fig:compose}
\end{figure}

The highest distortion required to embed a metric in $\set{M}\cup \cN$ into HSTs is a lower bound on the minimum distortion to embed $M_\beta[\cN]$ into HSTs, and careful composition reveals that in fact this quantity is an upper bound as well (up to constant factors).

%% file: figures/compose.tex
\tikzset{every picture/.style={line width=0.75pt}} 

\begin{tikzpicture}[x=0.75pt,y=0.75pt,yscale=-1,xscale=1]

\draw   (203.28,186.75) .. controls (202.64,181.92) and (204.75,177.13) .. (208.73,174.42) .. controls (212.7,171.71) and (217.85,171.56) .. (221.97,174.03) .. controls (223.43,171.21) and (226.1,169.27) .. (229.18,168.79) .. controls (232.26,168.3) and (235.38,169.33) .. (237.6,171.57) .. controls (238.84,169.02) and (241.29,167.31) .. (244.06,167.04) .. controls (246.84,166.77) and (249.55,167.98) .. (251.24,170.25) .. controls (253.49,167.55) and (257.07,166.41) .. (260.42,167.33) .. controls (263.78,168.25) and (266.32,171.06) .. (266.93,174.55) .. controls (269.69,175.31) and (271.98,177.27) .. (273.22,179.9) .. controls (274.47,182.53) and (274.53,185.58) .. (273.41,188.27) .. controls (276.12,191.88) and (276.75,196.68) .. (275.07,200.89) .. controls (273.39,205.11) and (269.65,208.09) .. (265.24,208.74) .. controls (265.21,212.69) and (263.09,216.32) .. (259.7,218.22) .. controls (256.3,220.12) and (252.16,220.01) .. (248.88,217.91) .. controls (247.48,222.65) and (243.54,226.13) .. (238.76,226.86) .. controls (233.99,227.59) and (229.23,225.43) .. (226.55,221.32) .. controls (223.26,223.34) and (219.31,223.93) .. (215.59,222.94) .. controls (211.88,221.95) and (208.71,219.46) .. (206.8,216.05) .. controls (203.44,216.45) and (200.19,214.67) .. (198.67,211.59) .. controls (197.14,208.51) and (197.66,204.79) .. (199.98,202.27) .. controls (196.98,200.47) and (195.45,196.89) .. (196.18,193.4) .. controls (196.92,189.92) and (199.76,187.31) .. (203.21,186.94) ; \draw   (199.98,202.27) .. controls (201.39,203.13) and (203.03,203.51) .. (204.66,203.38)(206.8,216.05) .. controls (207.51,215.96) and (208.2,215.79) .. (208.85,215.52)(226.55,221.32) .. controls (226.05,220.56) and (225.64,219.75) .. (225.31,218.9)(248.88,217.91) .. controls (249.13,217.05) and (249.3,216.16) .. (249.37,215.26)(265.24,208.74) .. controls (265.28,204.53) and (262.94,200.67) .. (259.23,198.83)(273.41,188.27) .. controls (272.81,189.7) and (271.89,190.97) .. (270.73,191.98)(266.93,174.55) .. controls (267.04,175.12) and (267.08,175.71) .. (267.07,176.3)(251.24,170.25) .. controls (250.68,170.92) and (250.22,171.68) .. (249.87,172.49)(237.6,171.57) .. controls (237.3,172.18) and (237.08,172.83) .. (236.93,173.5)(221.97,174.03) .. controls (222.84,174.55) and (223.65,175.18) .. (224.37,175.9)(203.28,186.75) .. controls (203.37,187.42) and (203.51,188.08) .. (203.7,188.72) ;
\draw   (297.28,46.75) .. controls (296.64,41.92) and (298.75,37.13) .. (302.73,34.42) .. controls (306.7,31.71) and (311.85,31.56) .. (315.97,34.03) .. controls (317.43,31.21) and (320.1,29.27) .. (323.18,28.79) .. controls (326.26,28.3) and (329.38,29.33) .. (331.6,31.57) .. controls (332.84,29.02) and (335.29,27.31) .. (338.06,27.04) .. controls (340.84,26.77) and (343.55,27.98) .. (345.24,30.25) .. controls (347.49,27.55) and (351.07,26.41) .. (354.42,27.33) .. controls (357.78,28.25) and (360.32,31.06) .. (360.93,34.55) .. controls (363.69,35.31) and (365.98,37.27) .. (367.22,39.9) .. controls (368.47,42.53) and (368.53,45.58) .. (367.41,48.27) .. controls (370.12,51.88) and (370.75,56.68) .. (369.07,60.89) .. controls (367.39,65.11) and (363.65,68.09) .. (359.24,68.74) .. controls (359.21,72.69) and (357.09,76.32) .. (353.7,78.22) .. controls (350.3,80.12) and (346.16,80.01) .. (342.88,77.91) .. controls (341.48,82.65) and (337.54,86.13) .. (332.76,86.86) .. controls (327.99,87.59) and (323.23,85.43) .. (320.55,81.32) .. controls (317.26,83.34) and (313.31,83.93) .. (309.59,82.94) .. controls (305.88,81.95) and (302.71,79.46) .. (300.8,76.05) .. controls (297.44,76.45) and (294.19,74.67) .. (292.67,71.59) .. controls (291.14,68.51) and (291.66,64.79) .. (293.98,62.27) .. controls (290.98,60.47) and (289.45,56.89) .. (290.18,53.4) .. controls (290.92,49.92) and (293.76,47.31) .. (297.21,46.94) ; \draw   (293.98,62.27) .. controls (295.39,63.13) and (297.03,63.51) .. (298.66,63.38)(300.8,76.05) .. controls (301.51,75.96) and (302.2,75.79) .. (302.85,75.52)(320.55,81.32) .. controls (320.05,80.56) and (319.64,79.75) .. (319.31,78.9)(342.88,77.91) .. controls (343.13,77.05) and (343.3,76.16) .. (343.37,75.26)(359.24,68.74) .. controls (359.28,64.53) and (356.94,60.67) .. (353.23,58.83)(367.41,48.27) .. controls (366.81,49.7) and (365.89,50.97) .. (364.73,51.98)(360.93,34.55) .. controls (361.04,35.12) and (361.08,35.71) .. (361.07,36.3)(345.24,30.25) .. controls (344.68,30.92) and (344.22,31.68) .. (343.87,32.49)(331.6,31.57) .. controls (331.3,32.18) and (331.08,32.83) .. (330.93,33.5)(315.97,34.03) .. controls (316.84,34.55) and (317.65,35.18) .. (318.37,35.9)(297.28,46.75) .. controls (297.37,47.42) and (297.51,48.08) .. (297.7,48.72) ;
\draw   (404.28,187.75) .. controls (403.64,182.92) and (405.75,178.13) .. (409.73,175.42) .. controls (413.7,172.71) and (418.85,172.56) .. (422.97,175.03) .. controls (424.43,172.21) and (427.1,170.27) .. (430.18,169.79) .. controls (433.26,169.3) and (436.38,170.33) .. (438.6,172.57) .. controls (439.84,170.02) and (442.29,168.31) .. (445.06,168.04) .. controls (447.84,167.77) and (450.55,168.98) .. (452.24,171.25) .. controls (454.49,168.55) and (458.07,167.41) .. (461.42,168.33) .. controls (464.78,169.25) and (467.32,172.06) .. (467.93,175.55) .. controls (470.69,176.31) and (472.98,178.27) .. (474.22,180.9) .. controls (475.47,183.53) and (475.53,186.58) .. (474.41,189.27) .. controls (477.12,192.88) and (477.75,197.68) .. (476.07,201.89) .. controls (474.39,206.11) and (470.65,209.09) .. (466.24,209.74) .. controls (466.21,213.69) and (464.09,217.32) .. (460.7,219.22) .. controls (457.3,221.12) and (453.16,221.01) .. (449.88,218.91) .. controls (448.48,223.65) and (444.54,227.13) .. (439.76,227.86) .. controls (434.99,228.59) and (430.23,226.43) .. (427.55,222.32) .. controls (424.26,224.34) and (420.31,224.93) .. (416.59,223.94) .. controls (412.88,222.95) and (409.71,220.46) .. (407.8,217.05) .. controls (404.44,217.45) and (401.19,215.67) .. (399.67,212.59) .. controls (398.14,209.51) and (398.66,205.79) .. (400.98,203.27) .. controls (397.98,201.47) and (396.45,197.89) .. (397.18,194.4) .. controls (397.92,190.92) and (400.76,188.31) .. (404.21,187.94) ; \draw   (400.98,203.27) .. controls (402.39,204.13) and (404.03,204.51) .. (405.66,204.38)(407.8,217.05) .. controls (408.51,216.96) and (409.2,216.79) .. (409.85,216.52)(427.55,222.32) .. controls (427.05,221.56) and (426.64,220.75) .. (426.31,219.9)(449.88,218.91) .. controls (450.13,218.05) and (450.3,217.16) .. (450.37,216.26)(466.24,209.74) .. controls (466.28,205.53) and (463.94,201.67) .. (460.23,199.83)(474.41,189.27) .. controls (473.81,190.7) and (472.89,191.97) .. (471.73,192.98)(467.93,175.55) .. controls (468.04,176.12) and (468.08,176.71) .. (468.07,177.3)(452.24,171.25) .. controls (451.68,171.92) and (451.22,172.68) .. (450.87,173.49)(438.6,172.57) .. controls (438.3,173.18) and (438.08,173.83) .. (437.93,174.5)(422.97,175.03) .. controls (423.84,175.55) and (424.65,176.18) .. (425.37,176.9)(404.28,187.75) .. controls (404.37,188.42) and (404.51,189.08) .. (404.7,189.72) ;
\draw   (257.33,149.3) -- (247.28,124.7) -- (256.22,129.17) -- (271.89,97.87) -- (262.95,93.4) -- (288.67,86.7) -- (298.72,111.3) -- (289.78,106.83) -- (274.11,138.13) -- (283.05,142.6) -- cycle ;
\draw   (297.02,192.29) -- (313.77,171.66) -- (314.14,181.65) -- (349.12,180.36) -- (348.75,170.37) -- (366.98,189.71) -- (350.23,210.34) -- (349.86,200.35) -- (314.88,201.64) -- (315.25,211.63) -- cycle ;
\draw   (372.45,95.23) -- (399.02,95.73) -- (391.37,102.17) -- (413.92,128.94) -- (421.57,122.5) -- (417.55,148.77) -- (390.98,148.27) -- (398.63,141.83) -- (376.08,115.06) -- (368.43,121.5) -- cycle ;

\draw (220,181) node [anchor=north west][inner sep=0.75pt]  [font=\Large] [align=left] {$\displaystyle N_{x}$};
\draw (312,43) node [anchor=north west][inner sep=0.75pt]  [font=\Large] [align=left] {$\displaystyle N_{y}$};
\draw (428,183) node [anchor=north west][inner sep=0.75pt]  [font=\Large] [align=left] {$\displaystyle N_{z}$};
\draw (174,100) node [anchor=north west][inner sep=0.75pt]   [align=left] {$\displaystyle \beta \Delta \cdot d_{M}( x,y)$};
\draw (419,96) node [anchor=north west][inner sep=0.75pt]   [align=left] {$\displaystyle \beta \Delta \cdot d_{M}( y,z)$};
\draw (290,224) node [anchor=north west][inner sep=0.75pt]   [align=left] {$\displaystyle \beta \Delta \cdot d_{M}( x,z)$};
\draw (411,31) node [anchor=north west][inner sep=0.75pt]   [align=left] {$\displaystyle \Delta :=\max_{w\in \{x,y,z\}}( diameter( N_{w}))$};

\end{tikzpicture}

%% file: buyatbulk.tex
\subsection{Buy-at-bulk\label{sec:buy}}

In the (uniform) buy-at-bulk problem,\footnote{We assume the metric version of the problem rather than the graph version, as Awerbuch et al. \cite{awerbuch1997} note that they are equivalent.} we are given a metric space $\cX$ 
as well as a set of demand pairs $(s_i,t_i)$, each with a demand $d_i$ and a sub-additive cost function $g$. (We will assume wlog that $g(1)=1$, as this is always possible by rescaling.) The goal is to select a single path $P_i$ for each demand pair in a way that minimizes $\sum_{e}\ell_eg(\sum_{i:e\in P_i}d_i)$, where $\ell_e$ is the length of edge $e$ and all the sum is over all edges. This has natural applications in many areas of operation research.

A key result of Awerbuch et al. \cite{awerbuch1997} states that if $\cD$ is a probabilistic tree embedding of metric $(X,d)$ with distortion $c$, then the optimal cost of the buy-at-bulk problem on $(X,d)$ is at most a factor of $c$ less than the expected value of the optimal solution to the same problem on the embedded metric. The buy-at-bulk problem is trivial on trees, where there is only one path between each pair of nodes. Thus, using the results of Munagala et al. \cite{mungala2023} we can obtain an expected approximation ratio $O(c)$, where $c$ is the optimal distortion for probabilistically embedding the input into HSTs.

In this subsection, we will work to improve this approximation ratio for certain metrics by solving the problem separately on a set of requests only involving a set of points with a low distortion HST embedding and then on the remaining set of requests, which we seek to ensure has relatively small total cost. We formalize this idea in Algorithm \ref{alg:appli}.\footnote{Algorithm \ref{alg:appli} actually gives a more general framework for solving the buy-at-bulk {\em or} dial-a-ride problem, and as a result it asks for an oracle that solves the underlying problem on HSTs. For the buy-at-bulk problem on trees, the only feasible solution is optimal, so $\cA$ is trivial.  } To outline this approach, we must first define some terms that will help us set up a ``weighted outlier'' problem on the input metric space. 

In particular, consider  a buy-at-bulk instance $(\cX,\set{(s_i,t_i,d_i)}_{i=1}^m)$. Let $OPT_i$ be the cost of the optimal solution to the buy-at-bulk instance $(\cX, (s_i,t_i,d_i))$ (i.e. the same instance but with only demand pair $i$). This is easy to solve optimally, as we can simply pick a shortest $s_i$-$t_i$ path and purchase $d_i$ capacity on that path. We will consider $OPT_i$ to be the ``cost'' of demand pair $i$. In order to decide the ``cost'' of removing a node $x$ from $\cX$, we will sum the costs of all demand pairs that involve $x$. That is, let $R_x:=\set{i:x\in \set{s_i,t_i}}$, and we will set the weight of $x$ to be $w(x):=\sum_{i\in R_x}OPT_i$. Further, for any $K\subseteq X$, let $w(K)=\sum_{x\in K}w(x)$. This weighting of the points of $\cX$ sets up a weighted outlier problem for $\cX$. 

For a given metric $\cX$  with weights $w$ and any choice of $c\geq 1$, we will define $\sminoutval{c}:=\min_{K\subseteq X;(X\setminus K)\probembed{c}\cH}w(K)$ to be  the cost of the cheapest set we can leave out while still embedding the remaining metric into HSTs with at most $c$ distortion. We will also be interested in the size of such a cheap set, so we define $\sminout{c}:= |\arg\min_{K\subseteq X;(X\setminus K)\probembed{c}\cH}w(K)|$. In Claim \ref{claim:buy-correct}, we give a bound on the cost of Algorithm \ref{alg:appli}.  Note that for some constant $\eta$, if  $c=\eta\log n$,\footnote{In this section, we let $\eta$ be a constant such that for all $n$-point metric spaces $X$, $X\probembed{\eta \log n}\cH$. (Such a constant is guaranteed to exist by \cite{frt04}.) } we have $\sminoutval{c}=0$, so  Algorithm \ref{alg:appli} never obtains worse than an $O(\log n)$ approximation.

\input{alg_appli}

\begin{clm}\label{claim:buy-correct}
Let $(\cX,R=\set{(s_i,t_i,d_i)}_{i=1}^m)$ be an instance of the buy-at-bulk problem with an overall solution of cost at most $OPT$.

Then Algorithm \ref{alg:appli} on this instance (and with the trivial oracle $\cA$ for buy-at-bulk) returns a solution of cost at most 
\begin{align*}
    O(\min_{c}\set{c\cdot OPT + \sminoutval{c}\cdot \log(n)}).
\end{align*}
In fact, for any $c$ directly considered by the algorithm in line 6, the algorithm's cost is bounded by $O({c\cdot OPT + \sminoutval{c}\cdot \log(\sminout{c})}$.
\end{clm}

First we briefly justify why proving the second part of the statement is sufficient to prove the overall statement, which minimizes over all $c$ not just the $c_i$ considered by the algorithm. First, for all $c\geq \eta \log n=c_0$, $\sminoutval{c}=0$, so we need only justify this for $c<\eta \log n$.  Let $c$ be the value that minimizes $c\cdot OPT + \sminoutval{c}\log (n)$, and note that any $c_i>c$ has optimal outlier set value at most equal to that of $c$. As the algorithm will consider some value of $c_i$ that is bigger than $c$ and at most $\eta$ times its value (if $c\geq \eta/2$, it will  in fact consider a choice of $c_i$ that is at most twice this value), we will end up with the desired approximation.

\begin{proofof}{Claim \ref{claim:buy-correct}}
We show that if index $i$ is selected as the optimal solution by the algorithm, then the cost of the solution is at most $O(c_i\cdot OPT + \sminoutval{c_i}\log( \sminoutset{c_i}))$. Let $K_i$ be the outlier set from line 8 of Algorithm \ref{alg:appli}.

Let $S_1$ be the set of pairs with neither endpoint in $K_i$ and let $S_2$ be the remaining pairs. We know that the cost of the solution we obtain for pairs in $S_1$ is at most $c_i\cdot OPT$, as the tree embedding incurs distortion $c_i$ and we get an optimal solution for the resulting tree. 

Now consider the cost associated with our purchase of edges for pairs in $S_2$. The total cost of purchasing for each demand pair independently is at most $\sum_{i\in S_2}OPT_i$. (As the cost function is sub-additive, buying more capacity on a single edge never costs more than splitting that capacity into two parts and buying them separately.)  We know that $\sum_{i\in S_2}OPT_i=O(\sminoutval{c_i}\log(\sminoutset{c_i}))$ by Algorithm \ref{alg:round}'s guarantees on the outlier set cost.
\end{proofof}

\subsubsection*{An example where the approximation ratio improves}

Let $n$ be such that $\log n$ is an even integer dividing $n$. Let $M$ be an HST metric of size $n/\log n$ and diameter $\ell \geq \log^3 n$, and let $\cN=\set{N_x}_{x\in M}$ be a set of HST metrics, each of size $\log n$.  Let $\cN'=\set{N_x'}_{x\in M}$ be the same as $\cN$, but for one  $x\st\in M$, replace $N_{x\st}$ with the metric of an unweighted expander graph that does not embed well into HSTs.   Let $M\st:=M_{1/2}[\cN']$ be a metric of size $n$.

Consider an instance of the buy-at-bulk problem on $M\st$ defined as follows
\begin{itemize}
    \item For each $x\in M$, divide the nodes of $N_x'$ into pairs $(u,v)$, and make $((x,u),(x,v))$ a demand pair in $M\st$ with demand $1$. 
    
    \item Let $x_1,x_2\neq x\st$ be a pair of nodes in $M$ with distance $\ell$. (We assume $x\st$ is chosen such that such a pair exists.) Select some $u\in N_{x_1}$ and $v\in N_{x_2}$ and add $((x_1,u),(x_2,v))$ as a demand pair with demand $1$.  
\end{itemize}

Note that $M\st$ requires distortion $\Omega(\log \log n)$ to  probabilistically embed into HSTs (due to the presence of an expander graph of size $\log n$ as a submetric), so if we run the algorithm of Munagala et al. \cite{mungala2023} and buy paths on the resulting tree, we could potentially pay as much as $\Omega(\log \log n \cdot OPT)$, depending on which distances were distorted. 

Now consider the result of running Algorithm \ref{alg:appli} on this instance. Note that if we remove from $M\st$ the nodes corresponding to $N_{x\st}'$, we get a metric that embeds into HSTs with constant distortion. We also know that the outlier cost of the nodes corresponding to $N_{x\st}'$ is at most $\log^2 n$, as there are $O(\log n)$ demand pairs associated with these nodes; each such pair has demand $1$; nodes in $N_{x\st}'$ only participate in demand pairs with other nodes in $N_{x\st}'$; and the diameter of $N_{x\st}'$ is at most $\log n$ (due to being the path metric of an unweighted graph on $\log n$ nodes). Thus, Algorithm \ref{alg:appli} gives us a solution of cost at most $O(OPT+\log^3 n)=O(OPT)$, where we used the fact that $OPT\geq \ell \geq \log^3n$, as there is some demand pair with distance at least $\ell$.

%% file: alg_appli.tex
\begin{algorithm}
\caption{Approximation algorithm for dial-a-ride or buy-at-bulk}
\label{alg:appli}
\textbf{Input:}  metric $\cX$, request set $\set{(s_i,t_i)}_{i=1}^m$, approximation algorithm $\cA$ for buy-at-bulk or dial-a-ride on HSTs, integer $\ell$ if dial-a-ride and demands $\set{d_i}$ if buy-at-bulk \\
\textbf{Output:} A buy-at-bulk or dial-a-ride solution
\begin{algorithmic}[1]
\For{$x=1$ to $n$}
\State $R_x\leftarrow\set{i:x\in \set{s_i,t_i}}$
\State Let $OPT_i$ be the cost of an optimal solution that fulfills only the $i$th request
\State $w_x\leftarrow \sum_{i\in R_x} OPT_i$
\EndFor
\For{$i=0$ to $\log \log n$}
\State Let $c_i\leftarrow \eta \log n/2^i$ \Comment{here $\eta$ is the appropriate constant from \cite{frt04}}
\State Run the weighted version of Algorithm \ref{alg:round} with input distortion $c_i$ and $\epsilon=1$ 
\State Receive an outlier set $K_i$ of cost $w^i$ and an embedding $\alpha_i$ of distortion $2c_i$
\State Embed $X\setminus K_i$ into HSTs
\State Use $\cA$ to solve the relevant problem on requests $R_1:=\set{j:\set{s_j,t_j}\cap K_i=\emptyset}$ and nodes in $X\setminus K_i$
\State Let $sol_i^{1}$ be the solution obtained
\State Let $R_2:=\set{j:\set{s_j,t_j}\cap K_i\neq \emptyset}=R\setminus R_1$
\State Find a minimum distortion embedding of $X_i=K_i \cup\set{x:\exists y \in K_i \text{ s.t. }\set{x,y}\in R}$ into HSTs
\State Use $\cA$ to find a solution to the problem with node set $X_i$ and request set $R_2$
\State Let $sol_i^{2}$ be the solution obtained
\State Find the ``naive'' solution to the problem with node set $X_i$ and request set $R_2$ 
\State Let $sol_i^{3}$ be the solution obtained 
\State Let $sol_i$ be the solution obtained by combining solution $sol_i^{1}$ with either $sol_i^{2}$ or $sol_i^{3}$, whichever is cheaper
\EndFor 
\State Return the solution of lowest cost
\end{algorithmic}
\end{algorithm}

%% file: dial.tex
\subsection{Dial-a-ride}

In the dial-a-ride problem, we have a van of capacity $\ell$ and we are given a metric space $(X,d)$ as well as a set of transportation requests $\set{(s_i,t_i)}_{i=1}^m$ from $m$ people such that person $i$ wants the van to transport them from location $s_i$ to location $t_i$.  The goal of the problem is to find a schedule of pickups and drop-offs for the van such that:\footnote{Technically we should also require that the tour start and end at some specified ``depot'' location, but solving this relaxed version of the problem and then adding on beginning and ending drops from/to the depot only worsens the overall approximation ratio by a constant factor (as each of the new trips we add has cost at most $OPT$), so we will discuss this simplified version instead. }
\begin{enumerate}
    \item The number of passengers in the van is the number that have been picked up minus the number that have been dropped off, and it is always between $0$ and $\ell$,

    \item For any customer $i$, that customer is picked up at their starting location $s_i$ and dropped off later at their destination $t_i$, and

    \item The cost of the solution is the sum of distances traveled by the van between each consecutive stop on its schedule.
\end{enumerate}

This problem is NP-hard, as it generalizes the Traveling Salesman Problem \cite{charikar1998finite}. Note that if we serve only a single request $i$, the problem is easy and incurs cost $d(s_i,t_i)$ in the optimal solution. Thus, computing $OPT_i$ for each request $i$ is efficient, and we can run Algorithm \ref{alg:appli} for this problem. 

However, we will need to use a more complicated algorithm $\cA$ to solve the dial-a-ride problem on HSTs than we did for buy-at-bulk. In particular, we will consider an algorithm $\cA$ that gets approximation ratio $r_{n,m,\ell}$ for dial-a-ride instances on HSTs with $n$ nodes, $m$ requests, and $\ell$ van capacity. In Claim \ref{claim:dial-alg}, we consider the approximation ratio achieved by Algorithm \ref{alg:appli} with such an input $\cA$. Again note that for $c=\eta\cdot \log n$, $\sminoutval{c}=0$, so the worst-case approximation ratio is $O(r_{n,m,\ell})$.

\begin{clm}\label{claim:dial-alg}
Let $\cA$ be an approximation algorithm for the dial-a-ride problem on HSTs that achieves approximation ratio $r_{n,m,\ell}$ for instances of $n$ nodes, $m$ requests, and $\ell$ van capacity. 

Let $(\cX,\set{(s_i,t_i)}_{i=1}^m,\ell)$ be a dial-a-ride input such that $OPT$ is the cost of an optimal solution for this instance, $\Delta$ is the  diameter of $\cX$, and $\sminoutval{c}$ and $\sminoutset{c}$ are as in Section \ref{sec:buy}.
Then Algorithm \ref{alg:appli} achieves a feasible solution with expected cost at most 
\begin{align*}
    O(\min_c\set{r_{n,m,\ell}\cdot c\cdot OPT +  \Delta \cdot \sminoutval{c}\log n}).
\end{align*}
In fact, for any $c$ directly considered by the algorithm in line 6, the algorithm's cost is bounded by \\ $O({c\cdot OPT + \sminoutval{c} \cdot \Delta\cdot  \log (\sminout{c})})$.
\end{clm}

Note that this algorithm's performance is generally better on instances whose optimal cost is much bigger than $\Delta$ (such as instances in which the total number of requests is large compared to $\ell\Delta$). 

Charikar et al. \cite{charikar1998} give an algorithm $\cA$ for dial-a-ride on HSTs with $r_{n,m,\ell}=O(\sqrt \ell)$. Thus, we get an algorithm that produces a solution such that for any $c\geq 1$, the solution's cost is at most $O(\sqrt{\ell}\cdot c\cdot OPT + \Delta\cdot\sminoutval{c}\log n)$. If each location is involved in at most $t$ requests (i.e. the demand is ``spread out''), then the solution also has cost at most $O(\sqrt{\ell}\cdot (c+\log(t\cdot  \sminoutset{c}))\cdot OPT)$.

As in Claim \ref{claim:buy-correct}, to prove Claim \ref{claim:dial-alg}, we only need to prove the given bound for each value of $i$ tested by the algorithm.

\begin{proofof}{Claim \ref{claim:dial-alg}}
We will show a slightly stronger version of the claim by showing that for each $i$, the cost associated with the solution for $c_i$ is bounded by $O(r_{n,m,\ell}\cdot c_i\cdot OPT +  \sminoutval{c_i}\log\sminoutset{c_i}\cdot \Delta)$ and also by $O(r_{n,m,\ell}\cdot c_i +  r_{2|R_2^i|,m,\ell}\log(|R_2^i|)\cdot OPT)$, where $R_2^i$ is the set of requests that intersect $K_i$ in iteration $i$ of the algorithm. (Only the first of these two bounds will be necessary to show the claim.) 
For the rest of this proof we will fix a choice of $i$ and consider the solution the algorithm obtains for that choice of $i$ and we refer to $R_2^i$ as $R_2$.

Let $K_i$ be the outlier set from line 8 of Algorithm \ref{alg:appli}, let $R_1$ be the set of requests $(s_j,t_j)$ that contain a point in $K_i$, and recall that $R_2$ is the set of remaining requests. Note that each solution the algorithm finds is feasible, as we handle all requests in $R_1$ according to some feasible solution given by $\cA$, and then we handle all requests in $R_2$ according to another feasible solution from $\cA$ or by a naive algorithm for this set of requests.

We will divide the cost of the algorithm's solution into three sections: the cost of servicing requests in $R_1$, the cost of servicing requests in $R_2$, and the cost of moving the van from the last drop-off location in the $R_1$ solution to the first pickup location in the $R_2$ solution.

\begin{enumerate}
    \item First consider the cost of the solution we obtain for servicing requests in $R_1$. Our expected cost is at most $r_{n,m,\ell}$ times the optimum solution on the tree this solution is embedded in, and the expected distortion of that embedding is bounded by $2c_{i}$ (as we potentially lose a factor of $2$ by setting $\epsilon=1$). Thus, the total cost of this part of the solution is bounded by $2r_{n,m,\ell}\cdot c_{i}\cdot OPT$ (noting that the optimal solution for $R_1$ is no more than the optimal solution for $R$).

    \item Now consider the cost of the second part of the solution. From the analysis of the weighted version of Algorithm \ref{alg:round}, we know that the total cost of the outlier set is at most $O(\sminoutval{c_{i}}\log(\sminoutset{c_{i}})$ This means 
    \begin{align*}
        \sum_{i\in R_2}d(s_i,t_i) &\leq  \sum_{x\in K_{i\st}}\sum_{i\in R_x} d(s_i,t_i) \\
        &\leq \sum_{x\in K_{i\st}}w_x  \\
        &\leq O(\sminoutval{c_{i}}\log(\sminoutset{c_{i}}).
    \end{align*}

    The first line comes immediately from the definition of $R_x$ (the set of requests involving $x$) and $R_2$; the second comes from the definition of $w_x$; and the third comes from our bound on the outlier set value of Algorithm \ref{alg:round}.
    
    Consider a naive algorithm that simply picks an arbitrary order on the requests and services the pick-ups and drop-offs one at a time in that order. The cost of the $j$th pickup and drop-off is $d(s_j,t_j)$, and the cost of moving from $t_j$ to $s_{j+1}$ is $d(t_j,s_{j+1})\leq \Delta$. In total we service at most $\sminoutval{c_{i}}$ requests in $R_2$, as we assume each request has a cost at least $1$. Thus, the total cost incurred on this part of the solution is bounded by $O(\sminoutval{c_{i}}\log\sminoutset{c_{i}}+\sminoutval{c_{i}}\cdot \Delta)$. 

    Additionally, we know that the optimal cost on this set of requests is bounded by $OPT$, and the number of points we need to embed to get a solution is at most $2|R_2|$. Thus, we can use a standard Fakcharoenphol et al. \cite{frt04} embedding on the set of points involved in this request and we get a solution of value at most $O(r_{2|R_2|,m,\ell} \log |R_2|OPT)$.

    \item Finally, consider the cost of the algorithm moving from the last drop-off spot $t_j$ in $R_1$ to the first pickup spot $s_{j'}$ in $R_2$. This total cost is $d(t_j,s_{j'})$, but this is upper bounded by $OPT$, as at some point any optimal algorithm must visit both locations.
\end{enumerate}
\end{proofof}